\documentclass[acmsmall]{acmart}
\usepackage{amsmath}
\usepackage{amsfonts}
\usepackage{algorithm}
\usepackage{algorithmic}
\usepackage{graphicx}
\usepackage{textcomp}
\usepackage{xcolor}
\usepackage{subfiles}
\usepackage{multirow}
\usepackage{booktabs}
\usepackage{vcell}
\usepackage{enumitem}
\usepackage{xspace}
\usepackage{hyperref}
\usepackage{colortbl}
\usepackage{tikz}
\usepackage{listings, color}
\usepackage{makecell}
\usepackage{array}
\usepackage{simplebnf}
\usepackage{tikz}
\usetikzlibrary{arrows.meta}
\usetikzlibrary{positioning}
\usepackage{needspace}

\newcommand{\thl}[1]{\cellcolor[HTML]{EFEFEF}#1}

\usepackage{amsthm}
\theoremstyle{plain}

\theoremstyle{definition}
\newtheorem{definition}{Definition}
\newtheorem{problem}{Problem}[section]

\newcommand{\coedpilot}{CoEdPilot\xspace}

\newcommand{\keeper}{EditFlow\xspace}
\newcommand{\editflow}{EditFlow\xspace}

\usepackage[most]{tcolorbox}
\tcbuselibrary{listingsutf8}

\newcommand{\highlight}[2]{%
    \textcolor{black}{#1}%
}
\newtcbox{\kbdbox}{on line,
  box align=base,
  colback=gray!10,
  colframe=black!60,
  boxrule=0.5pt,
  arc=3pt,
  top=1pt,bottom=1pt,left=2pt,right=2pt,
  boxsep=0pt,
  fontupper=\ttfamily\footnotesize
}

\definecolor{ForestGreen}{RGB}{34,139,34}
\lstdefinestyle{pythonstyle}{
    language=Python,
    basicstyle=\ttfamily\footnotesize,
    keywordstyle=\color{blue},
    stringstyle=\color{red},
    commentstyle=\color{ForestGreen},
    numberstyle=\tiny\color{gray},
    frame=single,
    showstringspaces=false,
    breaklines=true,
    tabsize=2
}

\AtBeginDocument{%
  }

\setcopyright{cc}
\setcctype{by}
\acmDOI{10.1145/3798249}
\acmYear{2026}
\acmJournal{PACMPL}
\acmVolume{10}
\acmNumber{OOPSLA1}
\acmArticle{141}
\acmMonth{4}
\received{2025-10-10}
\received[accepted]{2026-02-17}

\begin{document}

\title[EditFlow: Benchmarking and Optimizing Code Edit Recommendation Systems ...]{EditFlow: Benchmarking and Optimizing Code Edit Recommendation Systems via Reconstruction of Developer Flows}

\author{Chenyan Liu}
\orcid{0009-0005-0554-4028}
\email{chenyan@u.nus.edu}
\affiliation{%
  \institution{Shanghai Jiao Tong University}
  \city{Shanghai}
  \country{China}
}
\affiliation{%
  \institution{National University of Singapore}
  \city{Singapore}
  \country{Singapore}
}

\author{Yun Lin}
\orcid{0000-0001-8255-0118}
\email{lin_yun@sjtu.edu.cn}
\authornote{Corresponding author.}
\affiliation{%
  \institution{Shanghai Jiao Tong University}
  \city{Shanghai}
  \country{China}
}

\author{Jiaxin Chang}
\orcid{0009-0001-0199-1059}
\email{cjx001234@sjtu.edu.cn}
\affiliation{%
  \institution{Shanghai Jiao Tong University}
  \city{Shanghai}
  \country{China}
}

\author{Jiawei Liu}
\orcid{0009-0001-9823-0538}
\email{amberwabi2003@sjtu.edu.cn}
\affiliation{%
  \institution{Shanghai Jiao Tong University}
  \city{Shanghai}
  \country{China}
}

\author{Binhang Qi}
\orcid{0000-0002-0828-5544}
\email{qibh@nus.edu.sg}
\affiliation{%
  \institution{National University of Singapore}
  \city{Singapore}
  \country{Singapore}
}

\author{Bo Jiang}
\orcid{0009-0000-1080-3278}
\email{jiangbo.jacob@bytedance.com}
\affiliation{%
  \institution{Bytedance Network Technology}
  \city{Beijing}
  \country{China}
}

\author{Zhiyong Huang}
\orcid{0000-0002-1931-7775}
\email{huangzy@comp.nus.edu.sg}
\affiliation{%
  \institution{National University of Singapore}
  \city{Singapore}
  \country{Singapore}
}

\author{Jin Song Dong}
\orcid{0000-0002-6512-8326}
\email{dcsdjs@nus.edu.sg}
\affiliation{%
  \institution{National University of Singapore}
  \city{Singapore}
  \country{Singapore}
}

\renewcommand{\shortauthors}{C. Liu, Y. Lin, J. Chang, J. Liu, B. Qi, B. Jiang, Z. Huang, and J. S. Dong}

\begin{abstract}
Large language models (LLMs) for code editing have achieved remarkable progress, yet recent empirical studies reveal a fundamental disconnect between \textit{technical accuracy} and \textit{developer productivity}. Despite their strong benchmark performance, developers complete tasks 19\% slower when using AI assistance, with over 68.81\% of recommendations disrupting their mental flow. This misalignment stems from the use of static commit snapshots that lack temporal information, causing models to optimize for end results rather than the incremental, context-sensitive steps that align with developers’ natural reasoning process.

To bridge this gap, we present \textbf{EditFlow}, which benchmarks and optimizes subsequent code edit recommendation systems through the reconstruction of developer editing flows. EditFlow addresses three key challenges. 
First, collecting edit-order data that reflects developers' flow is inherently difficult: manual annotation introduces prohibitive overhead, while development logs capture only single trajectories instead of all plausible editing flows. 
Second, benchmarking recommendation performance against developers’ ongoing editing flow requires a \textit{digital-twin-like simulation} that can faithfully simulate the editing process.
Third, existing heterogeneous systems vary drastically in scale and architecture, posing challenges for developing a unified optimization strategy that endows all models with \textit{mental-flow awareness} regardless of design or capability.

To overcome these challenges, we propose three tightly coupled components: 
(1) a \textbf{prompt auto-tuning mechanism} that learns an optimized prompt for inferring the relative order between two edits,
(2) a \textbf{digital twin} that replays reconstructed edit sequences to simulate developers’ editing process, and 
(3) \textbf{\keeper}, a \textbf{unified optimization strategy} that optimizes the flow continuity of subsequent edit suggestions based on developers' ongoing flow.

Evaluations across diverse benchmarks, including manually annotated commits, real-world industrial code, and open-source repositories, show that EditFlow improves order reconstruction accuracy by 63.81\%, reduces flow violations by over 75\%, and boosts recommendation precision by 66.99\%. A user study with 32 developers further demonstrates 25.11\% faster task completion and significantly higher perceived recommendation quality. 
To the best of our knowledge, EditFlow is the first to evaluate and optimize code edit recommendation systems from the perspective of developers’ mental flow, establishing \textit{flow-awareness} as a new dimension for advancing human-AI code collaboration.
\end{abstract}

%


\maketitle

%
\section{Introduction} \label{sec:introduction}
Large language models (LLMs) for code have achieved remarkable progress, 
supporting tasks from code completion \cite{izadi2024language} to documentation synthesis \cite{cai2024fly}. 
As these capabilities mature\footnote{ChatGPT-o1-mini achieves 92.4\% accuracy on HumanEval \cite{o1-mini}}, 
recent efforts focus on integrating LLMs into developers' real-time workflows through \textit{subsequent edit recommendation}: suggesting the next plausible code edit immediately after a developer performs one.
Tools such as Cursor~\cite{cursor}, Claude Code~\cite{claude-code}, and \coedpilot~\cite{code-edit-pilot} exemplify this shift by offering proactive, context-sensitive suggestions. 
These systems support an interactive loop where model outputs are accepted, modified, or dismissed by the developer,
with the ultimate goal of reducing effort during code construction and maintenance.

However, recent empirical evidence suggests that this promise is far from guaranteed.
The controlled trial conducted by Becker et al. \cite{becker2025measuring} shows that,
despite the strong standalone benchmark performance \cite{huang2024olympicarenamedalranksintelligent} of Claude models (Claude 3.5 Sonnet achieving 92.0\% accuracy on HumanEval \cite{claude35news} and Claude 3.7 Sonnet reaching 70.3\% on SWE-Bench Verified \cite{claude37news}), developers using Cursor Pro integrated with these models completed tasks \textbf{19\% slower} than when working unaided.
This striking result reveals a fundamental disconnect: \textbf{accuracy is not equivalent to productivity}.
\highlight{A key factor underlying this disconnect is developers' \emph{mental flow}.
Mental flow is a well-established psychological construct defined as a mental state in which a person performing an activity is fully immersed in a feeling of energized focus, full involvement, and enjoyment~\cite{csikszentmihalyi1990flow}.
It is also a core determinant of developer productivity in both academic and industrial frameworks~\cite{meyer2014software,forsgren2021space,noda2023devex}.
Empirical studies consistently show that maintaining uninterrupted flow yields substantial productivity gains~\cite{Eirini_GitHub},
while even brief interruptions incur disproportionate recovery costs~\cite{mark2008cost,sum2015analysis}.}{(C1,C13)}

\highlight{To examine whether this disconnect stems from disruptions to developers' mental flow, we conduct two complementary studies.
First, through an analysis of 50 real-world commits using Cursor and Claude Code, we show that 68.81\% of model recommendations disrupt developers' ongoing mental flow, including 8.83\% of suggestions that are technically correct but ill-timed.
Second, through a controlled user study, we demonstrate that mitigating such flow disruptions leads to a 25.11\% improvement in task efficiency.}{(C1,C13)}

We argue that such disruption of mental flow by existing code editing tools arises from limitations in how current code LLMs are trained.
Most training data consists of commit snapshots, which represent only the final state of code changes, 
with the edits' temporal information lost during the commit process.
As a result, LLMs learn from outcomes rather than the editing process itself, 
missing the temporal order that reflects developers' natural mental flow.
This causes a fundamental mismatch: while models optimize for end results, 
developers work through incremental, context-sensitive steps that follow logical dependencies and cognitive continuity.

\highlight{Given this fundamental mismatch, our core objective is to improve AI code editing assistants by aligning suggestions with developers' ongoing mental flow.
However, realizing this objective is not straightforward, as it requires modeling aspects of the editing process that are largely absent from existing systems.
In particular, it depends on the ability to infer cognitive order relations between edits, and its effectiveness must be validated through process-level evaluation rather than outcome-based metrics.
Specifically, we identify three interconnected challenges that impede the development of flow-aware systems.}{(C3)}

\begin{itemize}[leftmargin=*]

    \item \highlight{\textbf{Challenge 1. The lack of flow-grounded edit-order data.}
    Avoiding flow violations fundamentally requires understanding \emph{which edits should follow which others} under developers' natural reasoning processes.
    However, such data is difficult to obtain:
    manual annotation is costly and unscalable,
    development logs capture only a single realized trajectory,
    and LLM-based annotation often fails to generalize due to brittle prompts.}{(C3)}

    \item \highlight{\textbf{Challenge 2. The lack of process-level flow-aware evaluation.}
    Existing evaluation protocols for code editing assistants are predominantly outcome-oriented.
    For instance, systems are commonly evaluated using edit location precision and recall, content-level similarity metrics such as BLEU, or task completion measured by test passing.
    These metrics fail to capture the temporal appropriateness and cognitive continuity of edit suggestions, making it difficult to assess or improve flow-awareness across systems.}{(C3)}

    \item \highlight{\textbf{Challenge 3. The difficulty of unifying heterogeneous systems.}
    Existing code recommendation systems are highly heterogeneous,
    ranging from academic prototypes to closed-source IDE assistants and CLI-based tools.
    This heterogeneity, together with diverse interfaces, inference behaviours, and latency constraints,
    makes it difficult to consistently evaluate flow-related behaviours
    and to systematically reduce flow violations across different systems.}{(C3)}

\end{itemize}

\highlight{To address the aforementioned challenges, we propose three tightly integrated components.
(1) A \textbf{prompt auto-tuning strategy} that leverages a small human-annotated dataset of edit orderings to learn an optimized prompt, capable of reconstructing the correct edit order between any two edits.
(2) A \textbf{digital twin framework} that simulates developers' editing trajectories and their interactions with code recommendation systems, guided by the recovered edit order,
enabling faithful evaluation of whether system-generated recommendations align with developers' ongoing mental flow.
(3) Building on the above components, we propose our flow-alignment optimization solution, \textbf{\keeper}, a unified wrapper for existing recommendation systems, which leverages the learned prompt to assess flow coherence, filters out flow-breaking recommendations.
These components are interdependent: the prompt auto-tuning strategy (1) provides the core capability for both the digital twin (2) and \keeper (3), while the digital twin (2) enables offline validation of the optimization (3).}{(C3)}

We comprehensively evaluate EditFlow through four research questions. 
On the annotated dataset (\textbf{RQ1}), the auto-tuned prompt achieves 87.26\% accuracy, substantially outperforming zero/few-shot and hand-crafted baselines by 63.81\%. 
On a real-world industrial dataset (\textbf{RQ2}), it demonstrates strong robustness with only 30 flow violations, representing a reduction of over 75\% compared to the best-performing baseline, and exhibiting close alignment with practical editing behavior. 
For edit recommendation (\textbf{RQ3}), the flow-aware optimization boosts precision across state-of-the-art systems (\textit{Cursor}, \textit{Claude Code}, and \textit{CoEdPilot}), yielding an average precision improvement of 66.99\%. 
Finally, a user study (\textbf{RQ4}) with 32 participants over 3 editing tasks shows 25.11\% faster task completion and higher perceived recommendation quality.


\needspace{3\baselineskip}
We summarize our main contributions are as follows:
\begin{itemize}[leftmargin=*]
    \item \highlight{To the best of our knowledge, we are the first to apply the concept of \textbf{mental flow in code editing}, and propose \textbf{\keeper}, a post-processing solution that enables \textbf{flow-aligned edit recommendation}.}{(C3)}

    \item \highlight{To operationalize mental flow, we propose a \textbf{prompt auto-tuning strategy} that learns an optimized prompt to infer edit orders, enabling accurate recovery of editing sequences as a structured representation of developers' mental flow.}{(C3)}

    \item \highlight{To evaluate flow alignment at the process level, we design a \textbf{process-based digital twin evaluation framework} that simulates realistic editing trajectories guided by the recovered mental flow graph, enabling flow-aware assessment of heterogeneous recommendation systems.}{(C3)}
    
    \item We demonstrate that \keeper substantially improves \textbf{flow-aligned recommendation precision} and \textbf{developer productivity} through extensive experiments and a controlled user study.
    
    \item We implement \keeper as a \textbf{VS Code extension} that provides visualized flow simulation and flow-aware recommendation filtering, with supplementary materials (e.g., demonstration videos, datasets and the learned prompt) available at our anonymous website~\cite{homepage}.
\end{itemize}

Building on these contributions, the remainder of the paper proceeds as follows.
We begin with a motivating commit example that illustrates how existing code editing tools violate developers’ mental flow during the editing process (\autoref{sec:motivation}).
We then formalize mental flow as pairwise edit order relations and define the problem of flow-aligned edit recommendation (\autoref{sec:problem_formulation}).
To support process-level assessment, we introduce a digital-twin-based evaluation framework with corresponding metrics (\autoref{sec:framework}).
Using this framework, we conduct an empirical validation on real-world commits to quantify the prevalence and types of flow-violating recommendations in existing systems (\autoref{sec:empirical_results}).
Finally, we present the design and implementation of \keeper, including prompt auto-tuning for edit order recovery and flow-aware recommendation filtering, and evaluate its effectiveness through extensive experiments and a controlled user study (\autoref{sec:methodology} \& \autoref{sec:experiment}).

\section{Motivating Example} \label{sec:motivation}
\begin{figure*}[ht]
    \centering
    \includegraphics[width=0.85\linewidth]{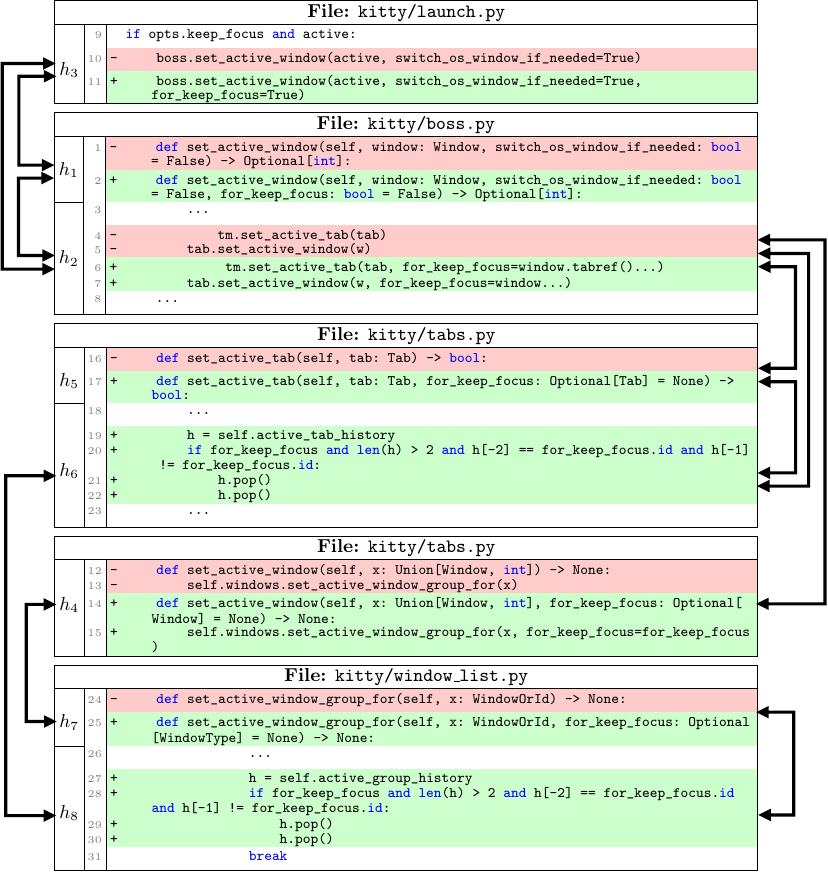}
    \caption{Motivating example: Edit partial order graph}
    \label{fig:motivation}
    \vspace{-5pt}
\end{figure*}

\autoref{fig:motivation} shows a real-world commit from project \texttt{kovidgoyal/kitty}\footnote{\url{https://github.com/kovidgoyal/kitty/commit/c4c62c15}},
which contains 8 edit hunks (denoted by $h_i$) across 4 files. 
Hunks such as $h_1 \& h_2$, $h_5 \& h_6$, and $h_7 \& h_8$ are grouped as paired edits, representing the function signature and corresponding implementation body.
An edge from $h_i$ to $h_j$ represents a partial order that users may proceed to $h_j$ after completing $h_i$. All edges in the graph are bidirectional.
This commit demonstrates a typical cascading change: 
adding a new parameter \texttt{for\_keep\_focus} to a function call at $h_3$,
which creates a ripple effect throughout the call chain.
The change first propagates to the target function's signature ($h_1$) and implementation ($h_2$), 
then branches into two parallel paths updating dependent functions: 
one path handles \texttt{set\_active\_window()} modifications ($h_2 \rightarrow h_4 \rightarrow h_7, h_8$), 
while the other addresses \texttt{set\_active\_tab()} changes ($h_2 \rightarrow h_5, h_6$).


For a set of edit hunks $\mathbb{H} = \{h_1, h_2, ..., h_8\}$ that together achieve an editing goal (adding the \texttt{for\_keep\_focus} parameter), 
the theoretical number of possible recommendation sequences is $|\mathbb{H}|! = 40,320$. 
However, the cognitively feasible sequences represent only a tiny fraction of this vast space. 
For example, after completing $h_3$, an AI assistant faces 7 possible immediate recommendations for the next edit, 
yet only $h_1$ and $h_2$ represent cognitively coherent choices that preserve the developer's mental flow continuity.
Beyond treating all remaining edits as equally viable, current systems compound this challenge by introducing false positives that further dilute the precision of flow-coherent suggestions. 
Users must bear additional cognitive overhead to parse and validate irrelevant recommendations, 
ultimately causing the assistance to reduce rather than enhance productivity.

To validate this analysis, we replicated the editing process in leading AI coding tools. 
After completing $h_3$, Claude Code immediately jumped to recommending $h_5$, skipping the natural intermediate steps $h_1$ and $h_2$ entirely. 
When we then completed the cognitively natural flow $h_3 \rightarrow h_1 \rightarrow h_2$ and queried Cursor for subsequent recommendations, it correctly suggested $h_5$ but failed in other aspects of sequencing. 
Cursor recommended reverting the user-completed edit $h_2$, creating cognitive dissonance by contradicting a previously applied change.

\highlight{Taken together, these failures highlight the urgent need for \emph{flow-aware optimization} (\autoref{prob:optimization}) to filter flow-violating suggestions.
Achieving this requires (1) inferring pairwise cognitive order relations (\autoref{prob:recovery}), and (2) a process-level evaluation framework to assess flow alignment throughout the editing process (\autoref{prob:evaluation}).
We address these needs through a unified pipeline. We first develop a prompt auto-tuning strategy for edit order recovery, then leverage the recovered flow graph to drive a digital twin framework for process-level evaluation, and finally build \keeper as a post-processing wrapper that filters and re-ranks recommendations across heterogeneous coding assistants.}{(C3)}

\section{Problem Formulation}\label{sec:problem_formulation}
\highlight{As established in earlier sections, existing coding assistants often disrupt developers' mental flow,
motivating us to formulate developers' mental flow as perceived order relations among edit hunks.
This formulation enables reasoning about which edits constitute cognitively coherent subsequent steps under a given editing context, forming the basis for flow-aware recommendation.}{(C3)}

\subsection{Preliminaries}

\begin{definition}[Edit Hunk]
\label{def:edit-hunk}
An \emph{edit hunk} $h$ is a tuple $(f, \ell_{\text{start}}, \ell_{\text{end}}, c_{\text{pre}}, c_{\text{post}})$ where:
\begin{itemize}
    \item $f$ is the file path;
    \item $\ell_{\text{start}}, \ell_{\text{end}}$ are line numbers defining the edit location;
    \item $c_{\text{pre}}, c_{\text{post}} \in \Sigma^*$ are the code content before and after the edit, where $\Sigma^*$ denotes the set of all strings (including the empty string $\varepsilon$).
\end{itemize}
Let $\mathbb{H} = \{h_1, h_2, \ldots, h_n\}$ denote all edit hunks in a commit $C$.
\end{definition}

\begin{definition}[Pairwise Edit Order]
\label{def:pairwise-order}
\highlight{We use the term \emph{partial order} to describe developers' perceived edit precedence from a cognitive perspective.
Unlike the mathematical definition, our notion of partial order does \emph{not} require transitivity: mental flow captures \emph{local} continuity between adjacent steps, and skipping intermediate edits may incur a cognitive gap that disrupts rather than preserves flow.}{(C8)}

For any pair of edit hunks $(h_i, h_j) \in \mathbb{H} \times \mathbb{H}$, we define their ground-truth partial order relation as $\lambda(h_i, h_j) \in \mathcal{L}$, where: $\mathcal{L} = \{\prec, \succ, \sim, \perp\}$.

\textbf{Interpretation}:
\begin{itemize}
    \item $\lambda(h_i, h_j) = \prec$: After completing $h_i$, a developer can \emph{naturally infer} the need for $h_j$ based on cognitive continuity.  
    \highlight{\textit{Example:} If $h_i$ removes a code block (\texttt{cut}) from one location, and $h_j$ subsequently inserts the same block elsewhere (\texttt{paste}), then $h_i \prec h_j$.}{(C6)}
    \item $\lambda(h_i, h_j) = \succ$: Conversely, completing $h_j$ naturally suggests $h_i$.
    \item $\lambda(h_i, h_j) = \sim$: Either direction preserves cognitive flow (symmetric).  
    \highlight{\textit{Example:} If $h_i$ deletes a parameter in a function's definition or signature, and $h_j$ removes the corresponding argument at all call sites, then either order can be cognitively coherent, i.e., $h_i \sim h_j$.}{(C6)}
    \item $\lambda(h_i, h_j) = \perp$: No apparent cognitive connection.
\end{itemize}

\end{definition}
\highlight{For conciseness, we limit the number of partial-order edit pair examples in the paper and provide additional illustrative cases on our anonymous website \cite{homepage}.}{(C6)}

\noindent\textbf{Remark (Cognitive vs. Technical Order).} 
The relation $\lambda$ characterizes \emph{mental flow} rather than strict syntactic or semantic dependencies. 
Specifically, $\lambda(h_i, h_j) = \prec$ does not imply that $h_i$ must precede $h_j$ to maintain program correctness or prevent syntax errors. 
Instead, it indicates that completing $h_i$ naturally renders $h_j$ a cognitively coherent next step within the developer’s reasoning process. 
Likewise, $\lambda(h_i, h_j) = \sim$ denotes a bidirectional cognitive relation, 
in which either edit may conceptually give rise to the other. 
For example, if $h_i$ introduces a new function definition and $h_j$ adds a corresponding function call, 
both orderings are cognitively plausible: a developer might first define the function to be invoked later, 
or alternatively, begin by writing the call and subsequently implement the function to fulfill the intent. 
We do not dismiss the latter direction merely because performing $h_j$ first may lead to transient reference or import errors, 
as such temporary inconsistencies do not disrupt the continuity of the developer’s mental flow.

\noindent\textbf{Remark (Approximating Ground-Truth Order).}
The ground-truth relation $\lambda(h_i, h_j)$ reflects an intrinsic cognitive property. 
Although collecting edit logs may appear to be an intuitive way to obtain edit order, 
such logs can only reveal the immediate sequence between consecutive edits and, at best, capture one possible ordering among many. 
They cannot recover the full set of pairwise relations within the edit set. 
Therefore, in practice, we approximate $\lambda$ via human annotation or inference by large language models, as described in \autoref{sec:auto-tuning}. Regardless of the approach, all annotations aim to capture the underlying cognitive order between edits.

\begin{definition}[Mental Flow Graph] \label{def:flow-graph}
Given edit hunks $\mathbb{H}$ and pairwise order labels $\lambda$, 
the \emph{mental flow graph} is a directed graph $G = (\mathbb{H}, \mathcal{E})$, 
where the edge set $\mathcal{E}$ is defined as:
\[
\mathcal{E} =
\{(h_i, h_j) \mid \lambda(h_i, h_j) = \prec\}
\cup
\{(h_j, h_i) \mid \lambda(h_i, h_j) = \succ\}
\cup
\{(h_i, h_j), (h_j, h_i) \mid \lambda(h_i, h_j) = \sim\}.
\]
\end{definition}

\highlight{\begin{definition}[One-Hop Successor]
\label{def:successor}
Given $G = (\mathbb{H}, \mathcal{E})$ and a set of completed edits $H_{\text{prior}} \subseteq \mathbb{H}$,
an edit $h' \in \mathbb{H} \setminus H_{\text{prior}}$ is a \emph{one-hop successor} of $H_{\text{prior}}$
if:
\[
\exists h \in H_{\text{prior}} : (h, h') \in \mathcal{E}.
\]
We define the set of one-hop successors of $H_{\text{prior}}$ as:
\[
\text{Succ}(H_{\text{prior}}, G)
= \{\, h' \in \mathbb{H} \setminus H_{\text{prior}}
\mid \exists h \in H_{\text{prior}} : (h, h') \in \mathcal{E} \,\}.
\]
\end{definition}}{(C8)}

\subsection{Problem Statements}
\label{sec:problem-statement}

\begin{problem}[Mental Flow Order Recovery]
\label{prob:recovery}

\highlight{Given a commit with edit hunks $\mathbb{H} = \{h_1, \ldots, h_n\}$,
the problem is to design a recovery method $M$ that infers the partial order relation
$\lambda: \mathbb{H} \times \mathbb{H} \to \mathcal{L}$
(as in \autoref{def:pairwise-order}) between any two edits,
which captures developers' perceived precedence between edits.
By applying the recovered order relation to all edit pairs,
the mental flow graph $G = (\mathbb{H}, \mathcal{E})$
(as in \autoref{def:flow-graph})
can be induced accordingly,
and serve as the ground-truth input for \autoref{prob:evaluation}.
The recovery method $M$ is also reused in \autoref{prob:optimization} to guide flow-aware optimization.}{(C3)}
\end{problem}

\begin{problem}[Process-driven Flow Alignment Evaluation]
\label{prob:evaluation}

\highlight{Given a code edit recommendation system $\mathcal{S}$ and a commit
$(\mathbb{H}, G)$ with mental flow graph $G = (\mathbb{H}, \mathcal{E})$ recovered from solving \autoref{prob:recovery},
the problem is to evaluate the extent to which $\mathcal{S}$ produces
flow-aligned edit recommendations during an editing process.
Specifically, the editing process is viewed as a sequence of intermediate states,
each represented by the set of already-applied edits $H^{(t)}_{\mathrm{prior}}$,
with $H^{(0)}_{\mathrm{prior}} = \emptyset$ and $H^{(T)}_{\mathrm{prior}} = \mathbb{H}$.
At each state $t$, the system $\mathcal{S}$ generates a set of candidate edits
$H^{(t)}_{\mathrm{pred}}$ conditioned on $H^{(t)}_{\mathrm{prior}}$.
The goal is to quantify, across intermediate states,
whether and to what extent edits in $H^{(t)}_{\mathrm{pred}}$
are flow-continuous with some edit in $H^{(t)}_{\mathrm{prior}}$
according to the mental flow graph $G$.}{(C3)}
\end{problem}

\begin{problem}[Flow-aware Optimization] \label{prob:optimization}
\highlight{
Given an editing context characterized by a set of completed edits $H_{\text{prior}}$,
and a set of currently suggested edit $H_{\text{pred}}$ predicted by code edit recommendation system $\mathcal{S}$,
we seek a post-processing mechanism $\phi: 2^H \to 2^H$.
The objective of $\phi$ is to reduce flow-violating recommendations
and prioritize flow-coherent subsequent edits, such that:
\[
\text{Precision}(\phi(H_{\text{pred}})) > \text{Precision}(H_{\text{pred}}).
\]
We address this problem in \autoref{sec:optimization}, by designing a flow-aware post-processing strategy, 
utilizing edit order recovery method $M$ from \autoref{prob:recovery}.}{(C3)}
\end{problem}

\noindent\textbf{Problem Relationships.}
\highlight{The three problems directly address the three challenges identified in \autoref{sec:introduction}:
Challenge~1 (the lack of flow-grounded edit-order data) motivates \autoref{prob:recovery};
Challenge~2 (the lack of process-level flow-aware evaluation) motivates \autoref{prob:evaluation};
and Challenge~3 (the difficulty of unifying heterogeneous systems) motivates both \autoref{prob:evaluation} and \autoref{prob:optimization}, as both handle diverse systems.}{(C3)}

\highlight{Beyond this correspondence, the three problems form a layered structure, as shown in \autoref{fig:problem-relationships}:
\autoref{prob:optimization} is the ultimate application goal, and the other two problems provide the necessary infrastructure.
Specifically,
\autoref{prob:recovery} provides the foundation that supports both \autoref{prob:evaluation} and \autoref{prob:optimization}:
the recovered mental flow graph serves as the ground-truth edit order in \autoref{prob:evaluation}, guiding the simulation of the editing process.
Meanwhile, the same recovery method is leveraged as a core component of the flow-aware filtering mechanism in \autoref{prob:optimization}.
\autoref{prob:evaluation} serves as a digital twin for validating
the optimization in \autoref{prob:optimization}, enabling offline
assessment without real user studies.}{(C3)}

\begin{figure}[h]
\centering
\begin{tikzpicture}[
  font=\footnotesize,
  box/.style={rectangle, draw, minimum width=2.5cm, minimum height=1cm, align=center},
  arrow/.style={->, >=stealth, thick}
]

\node[box] (p3) at (0, 2) {\autoref{prob:optimization}\\Optimization\\\textit{\textbf{(application goal)}}};
\node[box] (p1) at (-3, 0) {\autoref{prob:recovery}\\Recovery};
\node[box] (p2) at (3, 0) {\autoref{prob:evaluation}\\Evaluation};

\draw[arrow] (p1) -- node[left, font=\small] {Recovery method $M$} (p3);
\draw[arrow] (p1) -- node[below, font=\small] {Mental flow graph $G$} (p2);
\draw[arrow] (p2) -- node[right, font=\small] {Validation} (p3);

\end{tikzpicture}
\caption{Relationships among the three problems.}
\label{fig:problem-relationships}
\end{figure}
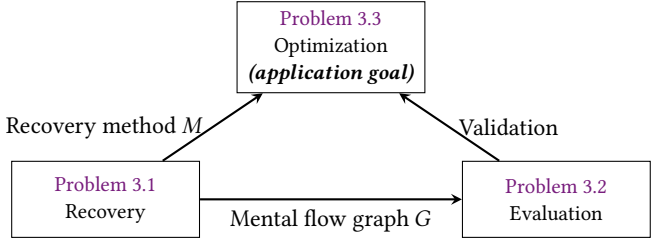

\section{Flow-Aware Evaluation Framework}\label{sec:framework}
\highlight{To enable process-driven, flow-aware evaluation,
as defined in \autoref{prob:evaluation},
we present a digital twin framework that simulates developers' editing processes with corresponding metrics.}{(C3)}

\subsection{Digital Twin} \label{sec:evaluation_pipeline}
To evaluate the flow-awareness of existing coding assistants at the progress level, 
we design a digital twin that simulates real-world editing processes guided by recovered developer flow.
For any git commit containing $n$ edit hunks $\mathbb{H}$, 
we construct an edit partial order graph $G$, as defined in \autoref{def:flow-graph}.
The simulation first checkout the commit to its pre-change version, 
where no edit has been applied yet ($H_{\text{prior}} = \emptyset$),
then randomly selects an edit hunk with minimum in-degree as the initial edit, 
applying it to the project.
\highlight{At each iteration, the digital twin interacts with the system under test (SUT) through a standard interface, supplying updated context as the simulation progresses, including:
(1) the list of previously simulated edits ($H_{\text{prior}} \subset \mathbb{H}$), (2) the access to the codebase at the current simulation progress, where previously simulated edits have all been applied to the project, and
(3) commit messages as edit descriptions.
SUTs are expected to return their recommended edits through this unified interface.}{(C10)}
This design ensures that the simulation system remains general and system-agnostic, 
with specific integration details for each evaluated SUT described in \autoref{sec:rq3_baseline}.
We evaluate the predicted edits through metrics defined in \autoref{subsec:metrics}.
If any of the recommended edit hunks are identified as a $\textsc{Keep}$ suggestion (i.e., it correctly aligns with a valid next edit in the flow graph $G$, see \autoref{def:flow-categories} for the formal definition),
the simulation system randomly selects a $\textsc{Keep}$ suggestion as the subsequent edit, applies it to the project, and updates the prior edits ($H_{\text{prior}} = H_{\text{prior}} \cup \{h\}$); 
otherwise, it selects from one-hop successor hunks ($h\in\text{Succ}(H_{\text{prior}}, G)$, as defined in \autoref{def:successor}).
This design assumes that, in real editing interactions, 
users would not apply incorrect edits to the project.
The simulation continues to request the next recommendation from SUT, 
until all ground-truth edit hunks have been simulated.
To better visualize this process, we implement this digital twin as a VS Code extension that 
capable of replaying this simulated editing process, 
with demonstration video available at our homepage \cite{homepage}.

This framework is used in two places throughout the paper.
First, in the empirical study (\autoref{sec:empirical_results}), it enables a small-scale quantification of whether existing coding assistants violate developers’ mental flow.
Second, in the main experiments (\autoref{sec:rq3}), it is used to evaluate assistant behaviour with and without \editflow, allowing us to assess the effectiveness of flow-aware optimization under the same evaluation protocol.

\subsection{Metrics}\label{subsec:metrics}

We design three categories of metrics to comprehensively evaluate coding assistants, capturing perspectives including flow alignment, recommendation correctness, and resource usage.

\subsubsection{Flow-Aware Metrics} \label{sec:flow-taxonomy}

Based on the flow graph, we classify predicted edits according to their alignment with the developer's current mental flow state.

\highlight{\begin{definition}[Flow Categories]\label{def:flow-categories}
Let $G = (\mathbb{H}, \mathcal{E})$ be a mental flow graph,
$H_{\text{prior}} \subseteq \mathbb{H}$ the sequence of completed edits,
and $h' \in H_{\text{pred}}$ denotes a predicted edit suggested by subsequent edit recommendation systems.
We define four mutually exclusive categories:
\begin{align}
\textsc{Keep}(h', H_{\text{prior}}, G) &\equiv h' \in \mathbb{H} \setminus H_{\text{prior}} \land h' \in \text{Succ}(H_{\text{prior}}, G) \label{eq:keep} \\
\textsc{Jump}(h', H_{\text{prior}}, G) &\equiv h' \in \mathbb{H} \setminus H_{\text{prior}} \land h' \notin \text{Succ}(H_{\text{prior}}, G) \label{eq:jump} \\
\textsc{Revert}(h', H_{\text{prior}}, G) &\equiv h' \in H_{\text{prior}} \label{eq:revert} \\
\textsc{Break}(h', H_{\text{prior}}, G) &\equiv h' \notin \mathbb{H} \label{eq:break}
\end{align}
\end{definition}}{(C8)}

\noindent Figure~\ref{fig:flow_category} illustrates these categories geometrically: A \textsc{Keep} edit maintains cognitive continuity by following an edge in $G$; 
a \textsc{Jump} edit targets a valid hunk but skips intermediate logical steps; 
a \textsc{Revert} edit suggests discarding an already-applied change; 
and a \textsc{Break} edit hallucinates content not present in the ground-truth commit.

\begin{figure}[h]
    \centering
    \includegraphics[width=0.85\linewidth]{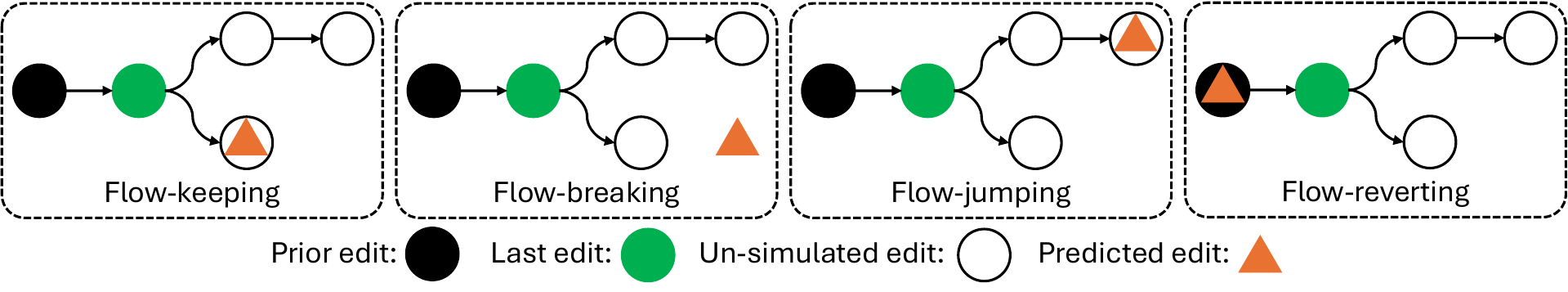}
    \vspace{-5pt}
    \caption{Flow categories of predicted edit}
    \label{fig:flow_category}
    \vspace{-5pt}
\end{figure}

\begin{lemma}[Partition Property]
\label{lem:partition}
For any predicted edit $h' \in \text{Edits}$ and fixed $H_{\text{prior}}$, $G$, exactly one of the four predicates in \autoref{def:flow-categories} holds.
\end{lemma}

\begin{proof}
By case analysis on $(h' \in H_{\text{pred}})$:
\begin{itemize}
\item If $h' \notin \mathbb{H}$: $\textsc{Break}(h')$ holds by Eq.~\eqref{eq:break}, and others fail since $h' \notin \mathbb{H} \implies h' \notin H_{\text{prior}}$ and $h' \notin \mathbb{H} \setminus H_{\text{prior}}$.
\item If $h' \in H_{\text{prior}}$: $\textsc{Revert}(h')$ holds by Eq.~\eqref{eq:revert}. Since $H_{\text{prior}} \subseteq \mathbb{H}$, we have $h' \in \mathbb{H}$ but $h' \notin \mathbb{H} \setminus H_{\text{prior}}$, so \textsc{Keep} and \textsc{Jump} fail. \textsc{Break} fails since $h' \in \mathbb{H}$.
\item If $h' \in \mathbb{H} \setminus H_{\text{prior}}$: By \autoref{def:successor}, either $h' \in \text{Succ}(H_{\text{prior}}, G)$ (then $\textsc{Keep}$ holds and $\textsc{Jump}$ fails) or $h' \notin \text{Succ}( H_{\text{prior}}, G)$ (then $\textsc{Jump}$ holds and $\textsc{Keep}$ fails). $\textsc{Revert}$ fails since $h' \notin H_{\text{prior}}$, and $\textsc{Break}$ fails since $h' \in \mathbb{H}$.
\end{itemize}
The cases are exhaustive (partition $\text{Edits}$ by membership in $\mathbb{H}$ and $H_{\text{prior}}$) and mutually exclusive (each case yields exactly one true predicate).
\end{proof}

\subsubsection{Flow-Independent Metrics}
To avoid potential bias introduced by the construction of flow graphs in flow-aware metrics,
we additionally define a set of flow-independent, ground-truth-based metrics.
Given a commit with a set of ground-truth edit hunks $\mathbb{H}$, a set of recommended next code edits from a single request $\mathbb{H}'$, and a set of all recommended edits during the entire simulation process $\mathbb{H}''$,
the flow-independent metrics evaluate recommendation quality based on traditional correctness criteria:
\begin{equation}
\text{Precision} = \frac{|\mathbb{H}' \cap \mathbb{H}|}{|\mathbb{H}'|}, \quad
\text{Recall} = \frac{|\mathbb{H}'' \cap \mathbb{H}|}{|\mathbb{H}|}, \quad
\text{F}_{0.5}=\frac{(1+0.5^2)\cdot \text{Precision}\cdot \text{Recall}}{0.5^2\cdot \text{Precision}+\text{Recall}}
\end{equation}

\highlight{\subsubsection{Resource Usage Metrics}
In addition to recommendation quality, we evaluate the resource overhead introduced by code editing assistants during real-time interaction.
We consider the following resource usage metrics:
\begin{itemize}[leftmargin=*]
    \item \textbf{Token Usage.}
    The total number of tokens consumed per recommendation, including both input tokens (prompt and context) and output tokens (generated edits or responses).
    This metric serves as a hardware-agnostic proxy for computational cost.
    \item \textbf{Latency.}
    The end-to-end wall-clock time required to generate a recommendation, measured from request issuance to response delivery.
    Latency reflects system responsiveness and directly affects interaction smoothness and mental flow.
    \item \textbf{Monetary Cost.}
    The estimated cost per recommendation is computed based on total token usage and the pricing scheme of the underlying model or service.
    For locally deployed open-source models, we report equivalent token-based cost to enable fair comparison.
\end{itemize}}{(C4,C16)}

\section{Empirical Validation} \label{sec:empirical_results}

To preliminarily validate our motivation, we perform an exploratory evaluation using the evaluation framework proposed in \autoref {sec:framework}.
We randomly selected 50 annotated Python commits from our training dataset with manually constructed edit partial order graph (details may refer to \autoref{sec:rq1_benchmark}), 
with Cursor CLI \cite{cursorcli} and Claude Code \cite{claude-code} as our SUTs.
The results are shown in \autoref{tab:empirical}.
We observe that the majority of predicted edits fall into the \textsc{Break} category, accounting for 55.48\% in Cursor and 51.23\% in Claude Code. 
In contrast, \textsc{Keep} edits, which truly align with the ongoing mental flow, 
represent only 28.23\% and 34.16\% respectively. 
The remaining predictions are largely distributed across \textsc{Jump} and \textsc{Revert} cases, 
both of which also disrupt the natural progression of edits.

Overall, these findings confirm our hypothesis: 
existing edit recommendation systems struggle to maintain mental flow alignment, 
with the majority of recommendations either irrelevant or disruptive. 
This misalignment explains why strong benchmark performance does not translate into productivity gains in real development workflows.

\begin{table}[ht]
\small
\centering
\caption{Empirical study of flow-violation in existing code edit recommendation system}
\vspace{-5pt}
\label{tab:empirical}
\begin{tabular}{ccccc} 
\hline
\multirow{2}{*}{\textbf{Baseline}}   & \multicolumn{4}{c}{\textbf{Flow categories (\%)}}               \\ 
\cline{2-5} & \textbf{Keep} & \textbf{Jump} & \textbf{Revert} & \textbf{Break} \\ 
\hline
Cursor  &  28.23 & 9.84 &  6.45 & 55.48 \\
\hline
Claude Code &  34.16 &  7.82 & 6.79 & 51.23 \\
\hline
\end{tabular}
\vspace{-5pt}
\end{table}

\section{Methodology}\label{sec:methodology}

\begin{figure*}[h]
    \centering
    \includegraphics[width=0.98\linewidth]{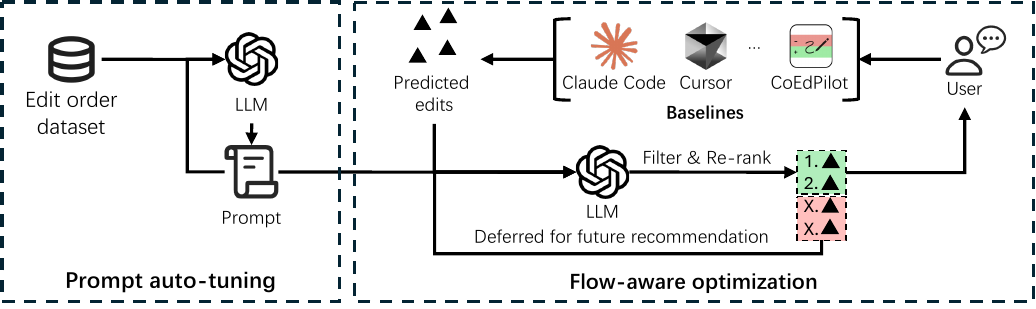}
    \caption{Overview of EditFlow}
    \label{fig:overview}
\end{figure*}

In this section, we present \keeper, our flow-alignment optimization solution for AI code editing assistants, as illustrated in \autoref{fig:overview}.
We first auto-tune an edit order recovery prompt from a human-labelled edit order dataset, 
and integrate such a prompt into the subsequent edit recommendation process, 
where we filter and re-rank edit suggestions produced by existing coding assistants based on their continuity with the mental flow reflected in prior edits.

\subsection{Edit Order Recovery} \label{sec:auto-tuning}
\highlight{To address \autoref{prob:recovery}, we propose a prompt-based approach that enables large language models to infer the pairwise cognitive order between code edits.}{(C3)}
Specifically, we aim to train and optimize task-specific prompts that guide the LLM to infer the partial order relation for any pair of edit hunks.

Since $\lambda$ is a latent cognitive property not directly observable 
from development histories, no existing training data is available.
We therefore manually construct an annotated dataset that approximates 
$\lambda$ by labelling pairwise order relations between edits within each commit.
Two authors conducted the annotation independently,
with access to the full repository and additional static analysis data, 
including dependency relations and the structural context of each edit 
(e.g., the enclosing class, function, or method call). 
The process followed two stages: 
(i) independent annotation of potential edit orders, and 
(ii) consensus resolution of disagreements through discussion.
On average, each commit required 20 minutes to annotate, resulting in a total annotation effort of 77 person-hours.

The resulting annotations define a set of pairwise relations, 
which are formalized as edges in an edit partial order graph. 
Each pair of hunks has a label from $\mathcal{L} = \{\prec, \succ, \sim, \perp\}$. 
Each hunk is represented using XML tags that encapsulate contextual information, 
with the edit itself expressed in Git diff format.
The contextual information includes the file path, the structural path of the edit, 
and any code elements that establish dependencies between the two hunks, as illustrated in \autoref{tab:edit_representation}.

\begin{table}[htbp]
\caption{Example of edit hunk representation}
\label{tab:edit_representation}
\vspace{-5pt}
\centering
\begin{lstlisting}[
    basicstyle=\ttfamily\scriptsize, 
    linewidth=0.98\columnwidth, 
    breaklines=true,
    xleftmargin=0.01\columnwidth
]
<file_path>kitty/launch.py</file_path>
<structural_path>
def launch(boss: Boss, ...)->Optional[Window]:
    boss.set_active_window(active, switch_os_window_if_needed=True)
</structural_path>
<code>
477 477        if opts.keep_focus and active:
478     -          boss.set_active_window(active, switch_os_window_if_needed=True)
    478 +          boss.set_active_window(active, switch_os_window_if_needed=True, for_keep
          _focus=True)
479 479            if opts.logo:
</code>
\end{lstlisting}
\vspace{-10pt}
\end{table}

\begin{algorithm}
\caption{Prompt Auto-Tuning for Edit Order Recovery}
\label{alg:prompt_tuning}
\begin{algorithmic}[1]
\STATE \textbf{Input:} Black-box large language model $LLM$; Training data $D_{tr} = \{(x_i,y_i)\}$, where $x_i = (h_i^0,h_i^1)$ is edit hunk pair, $y_i = \lambda(h_i^0,h_i^1) \in \{\prec, \succ, \sim, \perp\}$
\STATE \textbf{output:} Optimal prompt $p^*$
\STATE Randomly split $D_{tr}$ into batches: $B = \{b_1, \dots, b_k\}$
\STATE Obtain initial prompt candidates: $P=\{p_k|p_k := LLM(b_k)\}$
\STATE Select optimal prompt: $p^* := \arg\max_{p \in P}  \text{Accuracy}(p; D_{tr})$
\FOR{Epoch $e$ from 1 to $E$}
    \STATE $S^+=\{(x_i,y_i)|LLM(p^*, x_i)=y_i, (x_i,y_i)\in D_{tr}\}$ 
    \STATE $S^-=\{(x_i,y_i)|LLM(p^*, x_i)\neq y_i, (x_i,y_i)\in D_{tr}\}$     
    \STATE Re-split $D_{tr}$ into batches where each batch contains samples from both $S^+$ and $S^-$: $B=\{b_k|b_k\cap S^+\neq\emptyset, b_k\cap S^-\neq\emptyset\}$
    
    \FOR{each batch $b_k \in B$}
        \STATE Generate feedback: $F_k=\{f_i|f_i := LLM(p^*, b_k)\}$ 
        \STATE Integrate feedback to optimize prompt: $p_k := LLM(p^*, F_k)$ 
    \ENDFOR
    
    \STATE $p^* := \arg\max_{p \in \{p_1, \cdots, p_k\}}  \text{Accuracy}(p; D_{tr})$ 
\ENDFOR
\STATE \textbf{Return:} $p^*$
\end{algorithmic}
\end{algorithm}

Given this annotated dataset, we formalize the learning objective as identifying the optimal prompt $p^*$ that maximizes the predictive accuracy of the LLM on the pairwise annotation set:
\[
p^{*} = \arg\max_{p \in P} \text{Accuracy}(p; D_{tr}),
\]
where $P$ is the set of candidate prompts and $D_{tr}$ is the manually annotated training set.
As shown in \autoref{alg:prompt_tuning}, 
we split the dataset into $k$ batches and feed each complete batch to the LLM with only the response JSON schema, 
requiring the LLM to summarize an initial prompt for the edit order recovery task based on the input-output pairs $(x_i, y_i)$.
The algorithm iteratively refines prompts through feedback-driven optimization. 
In each epoch, we partition training data into correct predictions ($S^+$) and failed cases ($S^-$) 
using the current optimal prompt $p^*$. 
We then re-batch the dataset, 
ensuring each batch contains both successful and failed examples, 
providing contrastive learning signals.
For each batch $b_k$, 
we generate feedbacks $F_k$ by analyzing the prompt's performance, 
capturing why certain predictions succeeded while others failed. 
The LLM integrates this feedback to optimize prompt $p_k$, 
learning to distinguish between different edit order relationships. 
Finally, we evaluate all candidate prompts on the entire training set 
and select the one with the highest accuracy as the new optimal prompt $p^*$.
For the detailed learned prompt, please refer to our anonymous website \cite{homepage}.

\subsection{Flow-Aware Optimization}\label{sec:optimization}
\highlight{To optimize the flow alignment of edit recommendations (\autoref{prob:optimization}),
we implement \keeper as a wrapper that encapsulates existing sequential edit recommendation systems,
capable of improving the precision of edit suggestions by filtering out non-flow-keeping suggestions
and re-ranking them based on their relevance to maintain mental flow continuity.}{(C3)}
Given a code project $\mathbf{P}$, 
the sequence of prior edits $H_{\text{prior}} = \langle h_1, \ldots, h_n \rangle$, 
and the edit description $d$, 
\keeper queries the edit recommendation system to obtain a set of recommended edits, 
denoted as $H_{\text{pred}} = \{ h'_1, \ldots, h'_k \}$. 
\keeper then infers the pairwise order label between each predicted edit $h'_i$ and prior edit $h_n$ via $\text{LLM}(h_n, h'_i, p)$, 
where $p$ represents the auto-tuned edit order recovery prompt (as in \autoref{sec:auto-tuning}). 
The LLM outputs the inferred edit order label along with the log probability of each predicted token. 
A prediction of $\prec$ or $\sim$ indicates a potential edit order from the user's last edit $h_n$ to the predicted edit $h'_i$ ($h_n\rightarrow h'_i$ or $h_n\leftrightarrow h'_i$), 
whereas $\succ$ denotes an edit order $h'_i\rightarrow h_n$, which reverses the user's current mental flow, 
and $\perp$ suggests minimal relation between $h_n$ and $h'_i$.
Only edits predicted as either $\prec$ or $\sim$ are retained for subsequent processing.
\keeper re-ranks the remaining suggested edits according to the average log probability of the edit order label tokens predicted by $\text{LLM}(h_n, h'_i, p)$. 
A higher score indicates greater confidence in the predicted order label.
For edits that are predicted as $\succ$ or $\perp$, \keeper does not discard them permanently. Instead, these edits are deferred and may be reconsidered in future iterations when the user’s editing context evolves. As the sequence of prior edits grows, previously deferred edits can become flow-aligned and thus eligible for recommendation at a more appropriate time. This design ensures that \keeper prioritizes maintaining immediate mental flow continuity while preserving potentially useful edits for later recommendation.

We also implement \keeper as a VS Code extension, 
which wraps Cursor CLI, Claude Code and CoEdPilot as backends, 
filters and re-ranks their suggested edits before presenting them to users.
All edits are displayed in descending order of their average log probability, 
with each entry clickable to navigate to the corresponding file and view the associated edit diff.
For a demonstration of the interaction, please refer to our anonymous website \cite{homepage}.

\section{Experiment}\label{sec:experiment}
We evaluate our edit order recovery methods and various subsequent edit recommendation systems via the following research questions:

\begin{itemize}[leftmargin=*]
    \item \textbf{RQ1 (Effectiveness on annotated dataset, \autoref{sec:rq1})}: Can the learned partial order prompt perform well on our annotated dataset?
    \item \textbf{RQ2 (Alignment with practical data, \autoref{sec:rq2})}: How well does the partial order graph inferred by the learned prompt align with real-world editing data?
    \item \textbf{RQ3 (Performance of flow-aware optimization, \autoref{sec:rq3})}: Can our learned partial order prompt improve the suggestion precision of existing subsequent edit recommendation systems?
    \item \textbf{RQ4 (User study, \autoref{sec:rq4})}: Can our learned partial order prompt improve the real-world development productivity?
\end{itemize}

\subsection{RQ1 (Effectiveness on Annotated Dataset)} \label{sec:rq1}

\subsubsection{Setup} \label{sec:rq1_setup}
Prompt auto-tuning was performed and tested with \texttt{Claude-Sonnet-4-20250514}, using a maximum output length of 4096 tokens and a temperature of 0.7.
The auto-tuning underwent 5 epochs and a batch size of 32. 

\subsubsection{Benchmark} \label{sec:rq1_benchmark}
The annotated dataset consists of 100 commits from the 45 most-starred open-source GitHub Python repositories, comprising a total of 772 edit hunks and 1,747 directed edges.
From these, we constructed 2,030 training samples and 871 test samples.
To avoid intra-commit data leakage, we split the dataset at the commit level in a 7:3 ratio, ensuring that all samples from the same commit are assigned to the same split.
\needspace{3\baselineskip}
\subsubsection{Baseline} \label{sec:rq1_baseline}
We compare our auto-tuned prompt with 4 other baselines:
\begin{itemize}[leftmargin=*]
    \item \textbf{Zero-shot learning}: The LLM performs the task with only a single-sentence task description and the expected response schema, without any example or task-specific prompt tuning;
    \item \textbf{Few-shot learning}: The LLM is provided with 8 in-context examples, each consisting of a pair of edit hunks along with their edit order label, illustrating the edit order recovery task; the examples include four instances of each of the 2 possible labels.
    \item \textbf{Hand-crafted prompt}: A prompt manually designed based on the annotators' experience, summarizing key instructions for the edit order recovery task.
    \item \highlight{\textbf{DSPy} \cite{khattab2024dspy}: DSPy is a declarative framework for automated prompt optimization.
    We design a multi-step DSPy module that sequentially extracts edit information, analyzes dependency relations, and infers the final edit order via structured chain-of-thought reasoning.
    MIPROv2 optimizer \cite{opsahl2024optimizing} is adopted to refine the prompt based on human-labelled training data.}{(C9)}
\end{itemize}

\subsubsection{Metric}
We evaluated the quality of the generated prompts using accuracy, weighted precision, and F1 score (recall is equivalent to accuracy in this multi-class setting).

\subsubsection{Result}
\highlight{\autoref{tab:rq1} reports the performance of different prompting strategies on our annotated dataset.
Overall, the results show that the auto-tuned prompt substantially outperforms all baselines across accuracy, precision, and F1-score.
In particular, the auto-tuned prompt achieves an accuracy of 87.26\%, representing a 63.81\% relative improvement over the best-performing baseline (53.39\%).
This demonstrates the effectiveness of automated prompt optimization in capturing the underlying principles of edit ordering, surpassing other baseline strategies.
Notably, despite the common adoption of DSPy as a general-purpose prompt optimization framework,
its performance falls behind our auto-tuned approach, suggesting that generic example-based optimization is insufficient for capturing the fine-grained and rule-intensive nature of edit-order prediction.}{(C9)}

\begin{table}
\centering
\caption{The performance of the auto-tuned prompt on annotated dataset}
\label{tab:rq1}
\vspace{-5pt}
\begin{tabular}{lccc} 
\hline
\textbf{Method} & \textbf{Accuracy (\%)} & \textbf{Precision (\%)} & \textbf{F1-score (\%)}  \\ 
\hline
Zero-shot learning                                   & 50.63                  & 82.22                   & 61.67                   \\
Few-shot learning                                    & 47.42                  & 77.13                   & 57.96                   \\
Hand-crafted prompt                                  & 53.27                  & 80.96                   & 60.93                   \\
DSPy                                                 & 53.39                  & 62.44                   & 57.37\\
Auto-tuned prompt                                    & \textbf{87.26}         & \textbf{88.01}          & \textbf{87.54}         \\
\hline
\end{tabular}
\vspace{-5pt}
\end{table}

\subsection{RQ2 (Alignment with Practical Data)} \label{sec:rq2}

\subsubsection{Setup and Baseline}
This experiment adopted the same setup and baseline as depicted in \autoref{sec:rq1_setup} and \autoref{sec:rq1_baseline}.

\subsubsection{Benchmark}
To evaluate the quality of the auto-tuned prompt in the real world, we collected an edit order dataset from our industry collaborator, 
an international information technology company with over 60K employees. 
This real-world dataset is collected inside our collaborator, generated by its employees\footnote{The data collection obtained employee consent and underwent anonymization. 
All data is stored in the collaborator's internal systems throughout the process, industry-collected
and access to and analysis of the data during the research were conducted in a secure environment authorized by the collaborator. 
This dataset has not been disseminated outside the research team.}
This industry collected dataset consists of 500 commits from Jun. 2025 to Aug. 2025, containing 3,059 edit hunks.


We input any two edit hunks (e.g., $h_i, h_j$) within the same commit into the LLM, 
under the instruction of the auto-tuned prompt,
to infer their partial order.
Thereby, we construct a predicted partial order graph $G=(\mathbb{H},\mathcal{E})$, where $\mathbb{H}$ is the set of all edit hunks from the given commit and $\mathcal{E}$ is the set of inferred partial order edges. 

\subsubsection{Metric} 
We then evaluate the reliability of the predicted partial order graph, 
by measuring the number of forbidden partial order relationships inferred from the predicted partial order graph $G$,
that are contradicted by the ground-truth edit order sequence $S$.
The choice of a violation-based metric is motivated by the inherent ambiguity in evaluating partial order predictions against sequential ground truth. 
There may exist multiple natural and flow-keeping editing sequences 
for a commit with $n$ edit hunks.
However, practical data collection constraints limit us to observing a single editing sequence per commit, as it is infeasible to have developers repeatedly perform identical editing tasks to enumerate all valid orderings.
This asymmetry between the prediction space (partial orders) and 
the observation space (single gold sequence) renders traditional accuracy-based metrics problematic. 
Specifically, for any predicted partial order relationship that does not appear in the observed sequence, we cannot definitively classify it as incorrect,
since it may represent a valid alternative ordering that was simply not observed in our particular ground-truth instance.
To address this challenge, we adopt a violation-based approach that focuses on demonstrably incorrect predictions. 
While this approach is also affected by the single-gold-sequence limitation 
and provides an underestimation, 
it offers substantially higher reliability compared to accuracy-based metrics. 
The key advantage is that every detected violation represents a definitive error,
a constraint that is provably inconsistent with observed developer behaviour. 
The detailed algorithm to infer forbidden partial orders is available on our homepage \cite{homepage}.

\subsubsection{Result}

\begin{table}
\centering
\caption{Number of violations of the auto-tuned prompt and baselines on the real-world dataset}
\label{tab:rq2}
\vspace{-5pt}
\begin{tabular}{lc} 
\hline
\textbf{Method} & \textbf{\#Violation} \\ 
\hline
Zero-shot learning    & 195 \\
Few-shot learning     & 203 \\
Hand-crafted prompt   & 121 \\
Auto-tuned prompt     & \textbf{30} \\
\hline
\end{tabular}
\vspace{-5pt}
\end{table}

As shown in \autoref{tab:rq2}, 
the auto-tuned prompt significantly outperforms all baselines in terms of reducing violations. 
Specifically, while zero-shot and few-shot learning lead to 195 and 203 violations respectively, 
and the hand-crafted prompt reduces this number to 121, 
the auto-tuned prompt achieves only 30 violations.
This represents a reduction of over 75\% compared to the best-performing baseline, 
demonstrating the superior alignment of the auto-tuned prompt with real-world editing behaviors. 
These results confirm that prompt auto-tuning substantially enhances the reliability of partial order inference in practical development settings.

\subsection{RQ3 (Performance of Flow-aware Optimization)} \label{sec:rq3} 
We evaluate the subsequent edit recommendation performance of different baselines with and without our proposed flow-aware optimization technique.

\subsubsection{Experiment Design} \label{sec:rq2_experiment_desgin}
We adopt the digital twin for this experiment, with details described in \autoref{sec:evaluation_pipeline}.
For any git commit containing $n$ edit hunks $\mathbb{H}$, 
we construct an edit partial order graph $G$ using an auto-tuned prompt $p$.
We assume the inferred partial order graph as the ground-truth graph. 
\highlight{Notably, these graphs are not involved in \keeper's filtering process.
They are used to drive Digital Twin simulations by traversing the graph, ensuring that the simulated editing progress remains aligned with developers' mental flow, and to support the computation of flow-aware metrics.}{(C7)}

\subsubsection{Benchmark}
We adopt 2 benchmarks for this research question.
\begin{itemize}[leftmargin=*]
    \item \textbf{Large scale benchmark.} The simulation benchmark consists of 500 commits from 80 most-starred open-source GitHub Python repositories, and contains 3,584 hunks in total. For commits of this benchmark, we adopt the auto-tuned prompt and \texttt{Claude-Sonnet-4-20250514} to derive the edit partial order graph. 
    Commit selection criteria include: 
    (1) containing 5-10 edit hunks across at least 2 source files; 
    (2) involving real user authorship rather than automated systems;
    (3) excluding merge commits and filename changes; 
    and 
    (4) maintaining ASCII-only content with meaningful code modifications.
    \highlight{\item \textbf{Human annotated benchmark.} To avoid the concern that improvements on flow-aware metrics are a self-fulfilling consequence of EditFlow-generated order signals, we include our human-annotated benchmark with manually labeled partial orders.
    This benchmark consists of 25 Python commits, and more details may refer to \autoref{sec:rq1_benchmark}.}{(C5,C11,C18)}
\end{itemize}

\subsubsection{Baseline} \label{sec:rq3_baseline}
We selected Claude Code (Version 1.0.113), Cursor CLI (Version 2025.09.18-7ae6800), and CoEdPilot as our baseline.
For Cursor CLI and Claude Code, we relied on their default underlying models without manually specifying a particular model.
Each system requires different integration approaches with our simulation framework:
\begin{itemize}[leftmargin=*]
    \item \textbf{Claude Code and Cursor CLI:} Both the Anthropic Claude Code SDK \cite{claudecodesdk} and Cursor CLI \cite{cursorcli} support headless mode to avoid GUI interaction. 
    We supply the project path and structured requests with the same format, containing editing history and task specifications.
    The digital twin grants both Claude and Cursor agents comprehensive access to read, write files, and execute bash commands for project analysis.
    \item \textbf{CoEdPilot:} We employ CoEdPilot's discriminator model to select editing files, 
    then incorporate prior edits and edit descriptions into the locator and generator model inputs following CoEdPilot's specific formatting requirements.
\end{itemize}

For baselines without flow-aware optimization, we refer to this configuration as \textbf{Original}, where recommendations returned from baselines are directly evaluated.
For baselines with flow-aware optimization, we adopt our proposed \keeper to filter and re-rank returned suggestions before evaluation, we refer to this configuration as \textbf{w/ \keeper}.

\subsubsection{Metric} 
We evaluate flow-aware optimization via three categories of metrics, including flow alignment, recommendation correctness, and resource usage.
Details may refer to \autoref{subsec:metrics}.

\subsubsection{Result}
\highlight{For flow-aware metrics, as shown in \autoref{tab:rq3}, \editflow brings an average boost of 87.42\% to flow-keeping edits and an average drop of 22.42\% to flow-breaking, on the large-scale benchmark.
This improvement is not due to the LLM favouring its own derived editing order.
As shown in \autoref{tab:rq3_2}, on the human-annotated edit ordering benchmark,
\editflow achieves an average 106.58\% boost in flow-keeping edits and an average 19.15\% reduction in flow-breaking edits,
exhibiting a trend consistent with the large-scale benchmark results.}{(C5,C11,C18)}

\highlight{We further evaluate \editflow using precision, recall, and F0.5 as flow-independent metrics to provide a more objective assessment.
Across the two benchmarks, \editflow achieves an average 66.99\% improvement in precision, accompanied by a 7.09\% decrease in recall and an overall 46.64\% increase in F0.5, indicating a favourable trade-off that prioritizes recommendation correctness while maintaining stable coverage.}{(C5,C11,C18)}

\highlight{To assess the practical applicability of our approach, we further analyze its resource usage.
On average across three SUTs and two benchmarks, \editflow introduces an additional 1.71 seconds of latency, 6.58k tokens, and \$0.03 in cost per query. 
Both Claude and Cursor incur substantial latency from their agent-based workflows, where file inspection and tool invocation already dominate the end-to-end execution time. In comparison, the additional overhead introduced by \editflow accounts for only a small fraction of the overall latency.
Separately, prior studies suggest that short delays below approximately two seconds are generally not perceived as disruptive and are unlikely to break users' mental flow \cite{maslych2025mitigating}.
Moreover, these costs can be further reduced through engineering optimizations, such as adopting recent inference acceleration techniques \cite{cai2024medusa, wen2024speculative, butler2024pipeinfer},
leveraging more cost-efficient models as employed by Cursor.}{(C4,C16)}

In summary, the results demonstrate that flow-aware optimization leads to more flow-preserving recommendations, reducing disruptions while maintaining or improving overall effectiveness.

\begin{table}
\small
\centering
\caption{Performance of flow-aware optimization on large-scale benchmark}
\label{tab:rq3}
\vspace{-5pt}
\resizebox{\columnwidth}{!}{%
\begin{tabular}{clccccccccccc} 
\hline
\multirow{3}{*}{\textbf{Baseline}}     
& \multicolumn{1}{c}{\multirow{3}{*}{\textbf{Config}}} 
& \multicolumn{4}{c}{\textbf{Flow categories (\%)}}              
& \multirow{3}{*}{\begin{tabular}[c]{@{}c@{}}\textbf{Prec.}\\\textbf{(\%)}\end{tabular}}
& \multirow{3}{*}{\begin{tabular}[c]{@{}c@{}}\textbf{Rec.}\\\textbf{(\%)}\end{tabular}} 
& \multirow{3}{*}{\begin{tabular}[c]{@{}c@{}}\textbf{$\text{F}_{0.5}$}\\\textbf{(\%)}\end{tabular}}  
& \multicolumn{3}{c}{\textbf{Resource usage per query}} \\
\cline{3-6}\cline{10-12}
&  
& \textbf{Keep} & \textbf{Jump} & \textbf{Revert} & \textbf{Break} 
&  &  &  
& \begin{tabular}[c]{@{}c@{}}\textbf{Latency}\\\textbf{(s)}\end{tabular}
& \begin{tabular}[c]{@{}c@{}}\textbf{Tokens}\\\textbf{(K)}\end{tabular}
& \begin{tabular}[c]{@{}c@{}}\textbf{Cost}\\\textbf{(\$)}\end{tabular} \\  
\hline
\multirow{2}{*}{Cursor}  
& Original  
& 24.00 & 9.02 & 6.83 & 60.15 
& 33.02 & \thl{40.01} & 34.22
& 65.09 & 174.47 & 0.0621 \\
& w/ \keeper 
& \thl{38.49} & \thl{3.93} & \thl{6.45} & \thl{51.13} 
& \thl{42.42} & 37.97 & \thl{41.45}
& 67.12 & 180.68 & 0.0869 \\ 
\hline
\multirow{2}{*}{\begin{tabular}[c]{@{}c@{}}Claude \\Code\end{tabular}} 
& Original  
& 30.76 & 9.78 & 5.97 & 53.49 
& 40.54 & \thl{39.46} & 40.32
& 45.46 & 112.82 & 0.3609 \\
& w/ \keeper 
& \thl{46.89} & \thl{3.56} & \thl{4.03} & \thl{45.52} 
& \thl{50.45} & 36.15 & \thl{46.75}
& 46.70 & 117.15 & 0.3783 \\
\hline
\multirow{2}{*}{CoEdPilot}
& Original  
& 13.30 & \thl{1.48} & 2.31 & 82.91 
& 14.78 & \thl{28.08} & 16.33
& 6.49 & 43.59 & 0.0000 \\
& w/ \keeper 
& \thl{33.18} & 2.32 & \thl{1.93} & \thl{62.57} 
& \thl{35.50} & 25.68 & \thl{32.98}
& 8.38 & 54.80 & 0.0445 \\
\hline
\end{tabular}}
\vspace{-5pt}
\end{table}

\begin{table}
\small
\centering
\caption{Performance of flow-aware optimization on human-labelled benchmark}
\label{tab:rq3_2}
\vspace{-5pt}
\resizebox{\columnwidth}{!}{%
\begin{tabular}{clccccccccccc} 
\hline
\multirow{3}{*}{\textbf{Baseline}}     
& \multicolumn{1}{c}{\multirow{3}{*}{\textbf{Config}}} 
& \multicolumn{4}{c}{\textbf{Flow categories (\%)}}              
& \multirow{3}{*}{\begin{tabular}[c]{@{}c@{}}\textbf{Prec.}\\\textbf{(\%)}\end{tabular}}
& \multirow{3}{*}{\begin{tabular}[c]{@{}c@{}}\textbf{Rec.}\\\textbf{(\%)}\end{tabular}} 
& \multirow{3}{*}{\begin{tabular}[c]{@{}c@{}}\textbf{$\text{F}_{0.5}$}\\\textbf{(\%)}\end{tabular}}  
& \multicolumn{3}{c}{\textbf{Resource usage per query}} \\
\cline{3-6}\cline{10-12}
&  
& \textbf{Keep} & \textbf{Jump} & \textbf{Revert} & \textbf{Break} 
&  &  &  
& \begin{tabular}[c]{@{}c@{}}\textbf{Latency}\\\textbf{(s)}\end{tabular}
& \begin{tabular}[c]{@{}c@{}}\textbf{Tokens}\\\textbf{(K)}\end{tabular}
& \begin{tabular}[c]{@{}c@{}}\textbf{Cost}\\\textbf{(\$)}\end{tabular} \\  
\hline
\multirow{2}{*}{Cursor}  
& Original  
& 29.17 & 14.88 & 3.87 & 52.08
& 44.05 & \thl{45.95} & 44.41
& 51.29 & 195.81 & 0.0660 \\
& w/ \keeper 
& \thl{47.96} & \thl{05.58} & \thl{3.35} & \thl{43.12} 
& \thl{53.53} & 45.41 & \thl{51.68}
& 52.94 & 202.77 & 0.0926 \\ 
\hline
\multirow{2}{*}{\begin{tabular}[c]{@{}c@{}}Claude \\Code\end{tabular}} 
& Original  
& 29.55 & 10.12 & 6.88 & 53.44 
& 39.68 & \thl{38.92} & 39.52
& 50.53 & 114.62 & 0.3661 \\
& w/ \keeper 
& \thl{46.88} & \thl{2.08} & \thl{4.17} & \thl{46.88} 
& \thl{48.96} & 36.76 & \thl{45.91}
& 52.19 & 119.63 & 0.3853 \\
\hline
\multirow{2}{*}{CoEdPilot}           
& Original  
& 8.16 & \thl{1.84} & 3.95 & 86.05
& 10.00 & \thl{15.68} &  10.78
& 6.83 & 45.29 & 0.0000 \\
& w/ \keeper 
& \thl{23.61} & 2.78 & \thl{3.47} & \thl{70.14} 
& \thl{26.39} & 14.05 & \thl{22.45}
& 8.63 & 51.04 & 0.0223 \\
\hline
\end{tabular}}
\vspace{-5pt}
\end{table}

\subsection{RQ4 (User Study)} \label{sec:rq4} 
To validate \keeper performance in real-world usage, we design this user study with our implemented VS Code extension.
Further details, including demographic details and instrumentation-based analysis of user interactions, are available on our homepage \cite{homepage}.

\subsubsection{System Under Test}
We evaluate the effectiveness of our \keeper 
by applying it to two representative code editing systems, 
CoEdPilot and Claude Code. 
For each system, we compare two variants: (1) the original system without \keeper, and
(2) the optimized system with \keeper integrated. 
This setup allows us to isolate the contribution of \keeper by controlling for differences across systems and directly measuring its impact on editing recommendation quality in terms of flow-awareness.

\subsubsection{Participant} 
We recruited 32 participants from 2 universities.
All participants are required to complete a pre-study questionnaire to collect their background information, 
including educational level, programming proficiency, and prior experience with AI-assisted programming tools.
The participants were aged between 20 and 30 and enrolled in computer science programs, ranging from undergraduate to PhD level. 
They demonstrated substantial proficiency in Python programming, engaging in coding activities on average 4.5 days per week. 
Additionally, 90\% of participants reported prior experience with AI-assisted programming tools.

Participants using the original Claude Code constitute Control Group 1 (CG1), 
while those using Claude Code integrated with \keeper form Experiment Group 1 (EG1). 
Similarly, participants using the original \coedpilot comprise Control Group 2 (CG2),
and those using \coedpilot with \keeper constitute Experiment Group 2 (EG2).

\subsubsection{Setup}
Participants first complete a warm-up task that follows the same procedure as the three formal tasks. 
The warm-up has no time limit and is intended to familiarize them with the extension and study workflow.
For each task, we provide:
\begin{enumerate}[leftmargin=*]
    \item The complete repository;
    \item Detailed edit descriptions, with the necessary background knowledge, including code functionality and method interface;
    \item An initial edit that serves as a starting point, guiding participants toward completing the remaining edits;
    \item Test cases that participants can execute to verify whether their edits meet the expected outcomes. 
    A task is considered complete once all tests pass.
\end{enumerate}
At the end of the study, 
participants are required to share screen recordings throughout their completion of the three tasks, 
along with the instrumentation data collected during their interaction.

\subsubsection{Editing Task}
Participants are required to complete three real-world editing tasks derived from commits in top-starred GitHub repositories.
\highlight{The three commits are intentionally selected to span distinct edit-topology structures observed in real-world multi-location edits, allowing us to examine both the effectiveness and the boundaries of flow-aware optimization.}{(C12)}
Tasks are simplified by providing project-specific knowledge and removing distracting edits,
so that participants can complete them without requiring specialized background knowledge.
Each task is limited to 30 minutes.

\begin{itemize}[leftmargin=*]
    \item \textbf{Task 1 (Moderately easy)}: Originating from \cite{task1}, this commit adds a new argument \texttt{allow\_scrap\ ing} to \texttt{fetch\_url()}, controlling whether the system may retrieve source code from an external URL for bug localization upon a Sentry-detected error.
    In the call chain of \texttt{fetch\_url()}, there are syntax-based modifications across three files with a total of 8 changes. We select the implementation involving the concrete use of the \texttt{allow\_scraping} parameter as the initial edit, from which the remaining edits can be inferred by users.
    \item \textbf{Task 2 (Hard)}: Originating from \cite{task2}, this commit ensures that when \texttt{--keep-focus} is used, the tab and window history remain consistent, preserving the correct focus transition behaviour by removing non-user-initiated entries from the history stack. This update involves 8 changes across 4 files, with the update of calling function \texttt{boss.set\_active\_window()} as the initial edit.
    \item \textbf{Task 3 (Easy)}: Based on \cite{task3}, this commit focuses on refactoring a commonly used \texttt{if} condition into a standalone function. The change spans 14 modifications across 2 files. We treat the extracted function as the initial edit, and the user’s task is to locate the corresponding \texttt{if} condition instances and replace them with this function.
\end{itemize}

\subsubsection{Metric}
We assess the user study along both quantitative and qualitative dimensions. 
Quantitatively, we measure the \textbf{average time cost} for user efficiency, 
compute \textbf{Mann-Whitney U test p-values} using permutation testing (10,000 resamples) to determine statistical significance between independent groups ($p$ < 0.05 indicates significance), 
and calculate \textbf{effect sizes} ($r$) derived from the standardized U statistic to quantify the magnitude of differences ($r \geq$ 0.5 indicates a large effect). 
Qualitatively, we randomly sample screen recordings from each group and examine participants' editing behaviours in conjunction with instrumentation data.

\begin{table}
\centering
\caption{Performance of Claude Code with \keeper (EG1), \coedpilot with \keeper (EG2), Claude Code (CG1), and \coedpilot (CG2), time cost in minutes}
\label{tab:rq4}
\vspace{-5pt}
\resizebox{\textwidth}{!}{
\begin{tabular}{cccc|cccc|cccc|cccc} 
\hline
\textbf{EG1}                      & \textbf{T1}   & \textbf{T2}   & \textbf{T3}   & \textbf{CG1}                      & \textbf{T1}   & \textbf{T2}    & \textbf{T3}   & \textbf{EG2}                      & \textbf{T1}   & \textbf{T2}   & \textbf{T3}   & \textbf{CG2}                      & \textbf{T1}   & \textbf{T2}    & \textbf{T3}    \\ 
\hline
\textbf{P1}                       & 6.42          & 8.29          & 3.82          & \textbf{P9}                       & 7.67          & 11.13          & 3.05          & \textbf{P17}                      & 6.83          & 7.61          & 3.59          & \textbf{P25}                      & 12.76         & 14.57          & 4.84           \\
\textbf{P2}                       & 8.05          & 8.93          & 5.33          & \textbf{P10}                      & 6.99          & 10.32          & 4.59          & \textbf{P18}                      & 8.43          & 6.10          & 4.03          & \textbf{P26}                      & 9.15          & 12.27          & 4.57           \\
\textbf{P3}                       & 6.59          & 6.45          & 3.36          & \textbf{P11}                      & 7.07          & 15.79          & 4.93          & \textbf{P19}                      & 7.48          & 5.86          & 5.29          & \textbf{P27}                      & 8.22          & 11.56          & 4.73           \\
\textbf{P4}                       & 5.45          & 6.33          & 3.89          & \textbf{P12}                      & 8.69          & 12.77          & 3.98          & \textbf{P20}                      & 6.42          & 6.45          & 4.97          & \textbf{P28}                      & 7.28          & 15.88          & 3.64           \\
\textbf{P5}                       & 7.05          & 6.78          & 3.23          & \textbf{P13}                      & 8.62          & 15.78          & 5.66          & \textbf{P21}                      & 4.32          & 5.91          & 3.31          & \textbf{P29}                      & 8.84          & 12.31          & 5.57           \\
\textbf{P6}                       & 8.07          & 10.94         & 3.40          & \textbf{P14}                      & 7.49          & 12.33          & 3.79          & \textbf{P22}                      & 5.66          & 8.03          & 2.10          & \textbf{P30}                      & 6.92          & 11.12          & 4.33           \\
\textbf{P7}                       & 5.88          & 8.20          & 3.15          & \textbf{P15}                      & 7.93          & 9.37           & 4.47          & \textbf{P23}                      & 7.24          & 6.98          & 3.64          & \textbf{P31}                      & 8.81          & 14.58          & 4.93           \\
\textbf{P8}                       & 7.51          & 8.10          & 2.08          & \textbf{P16}                      & 5.64          & 11.77          & 2.62          & \textbf{P24}                      & 9.62          & 5.58          & 3.66          & \textbf{P32}                      & 17.00         & 12.39          & 2.64           \\ 
\hline
\multicolumn{1}{l}{\textbf{Avg.}} & \textbf{6.88} & \textbf{8.00} & \textbf{3.53} & \multicolumn{1}{l}{\textbf{Avg.}} & \textbf{7.51} & \textbf{12.41} & \textbf{4.14} & \multicolumn{1}{l}{\textbf{Avg.}} & \textbf{7.00} & \textbf{6.57} & \textbf{3.82} & \multicolumn{1}{l}{\textbf{Avg.}} & \textbf{9.87} & \textbf{13.09} & \textbf{4.41}  \\
\hline
\end{tabular}
}
\vspace{-5pt}
\end{table}

\subsubsection{Results and Analysis}
The performance of all participants is shown in \autoref{tab:rq4}, with the following analysis:

\begin{itemize}[leftmargin=*]
    \item \textbf{Task 1:} EG1 marginally outperforms CG1 without statistical significance ($p = 0.1966, r = 0.289$), whereas EG2 outperforms CG2 with statistical significance ($p = 0.0318, r = 0.525$)
    \item \textbf{Task 2:} The superiority of EG1 over CG1 is statistically significant ($p = 0.0004, r = 0.788$), and a consistent trend is also observed for EG2 compared with CG2 ($p = 0.0004, r = 0.840$).
    \item \textbf{Task 3:} EG1 marginally outperforms CG1, though the difference is not statistically significant ($p = 0.2186, r = 0.289$). A similar trend is observed when comparing EG2 with CG2 ($p = 0.2360, r = 0.289$).
\end{itemize}

\textbf{Why is \keeper’s advantage over the original system not evident in Task 1?}
Task 1 is of low-to-moderate difficulty, where most edits follow clear syntactic propagation patterns,
These patterns constrain the search space for plausible edits and make it easier for users to maintain a clear mental flow.
As a result, strong models such as Claude already produce fewer flow-breaking recommendations, leaving limited room for further improvement.
In contrast, weaker models like CoEdPilot are more prone to incorrect or mistimed suggestions, allowing \keeper to deliver statistically significant gains.
\highlight{These results indicate that \keeper is effective at narrowing the performance gap between large and small models in edit recommendation tasks.}{(C12)}

\textbf{Why does \keeper outperform both original systems in Task 2?} 
Task 2 is substantially more challenging,
requiring a deeper understanding of the codebase and permitting only a small set of valid edit sequences.
Moreover, it involves several functions with similar names: 
\texttt{set\_active\_tab()} (ground truth), \texttt{\_set\_active\_tab()} (distraction), and \texttt{set\_active\_tab\_idx()} (distraction), 
increases the likelihood of confusion and flow-violating recommendations from both Claude and CoEdPilot.
\highlight{In this setting, \keeper provides consistent and statistically significant improvements for both systems.}{(C12)}
Without \keeper, the original systems tend to produce flow-violating suggestions, imposing a high verification burden on users.
For example, user P31 had 13 locations where each verification took more than 10 seconds, 
including 4 locations that each took more than 20 seconds and 2 that exceeded 30 seconds.
If the user does not strictly verify the recommendations, the incorrect edits are usually only discovered during test execution, requiring substantial time for debugging (approximately 60\% of the total editing time for P13).

\textbf{Why does \keeper\ not demonstrate a statistically significant advantage in Task 3?}
Task 3 is a uniform refactoring task involving 14 highly similar edits, where an \textit{if} condition is consistently replaced by a call to a refactored function,
which introduce no functionality or semantic side effects that affect the surrounding code.
\highlight{By construction, this scenario induces an almost fully connected partial-order graph in which edits are largely interchangeable, making the task inherently order-insensitive.
As a result, nearly any edit ordering aligns with developers' mental flow, and both Claude and CoEdPilot already generate correct recommendations with few flow violations.
This task, therefore, serves as a boundary case where flow-aware optimization is not expected to yield substantial benefits, consistent with our observed results.}{(C12)}

\section{Failure Analysis and Metric Implications}
\label{sec:failure_analysis}

\highlight{While \editflow consistently improves multiple flow-aware and flow-agnostic metrics,
we observe several recurring scenarios in which \editflow itself fails to make correct flow-aware decisions.
In this section, we focus specifically on analyzing \editflow's failure modes.
For each failure category, we qualitatively explain the underlying mechanism in \editflow's design or assumptions that lead to incorrect decisions,
illustrate representative examples,
and discuss how such failures propagate to downstream evaluation metrics.}{(C2,C19)}

\subsection{Failure Mode I: False Rejection due to $k$-Context Sensitivity}
\highlight{\textbf{Mechanism.}
In this scenario, \keeper incorrectly filters out flow-aligned subsequent edit recommendations.
This failure arises when \keeper evaluates flow continuity using only the most recent edit, effectively assuming a 1-step editing context.
However, in practice, a valid subsequent edit may exhibit flow continuity with earlier edits rather than with the most recent edit alone.
When such broader contextual information lies outside the limited context window considered by \keeper, 
the candidate edit appears locally misaligned and is therefore incorrectly rejected, despite being consistent with the developer’s overall editing flow.
\definecolor{codegreen}{rgb}{0.0, 0.66, 0.30}
\definecolor{codered}{rgb}{1,0.05,0.05}
\definecolor{codegray}{rgb}{0.5,0.5,0.5}
\definecolor{codegray}{rgb}{0.5,0.5,0.5}
\definecolor{codeblue}{rgb}{0.0,0.0,0.5}
\definecolor{codered}{rgb}{0.6,0.0,0.0}
\definecolor{codepurple}{rgb}{0.58,0.0,0.82}
\definecolor{codeorange}{rgb}{1.0,0.5,0.0}
\definecolor{codecyan}{rgb}{0.0,0.6,0.6} 

\lstdefinestyle{lgeometry}{ xleftmargin = 20pt, xrightmargin = 0pt, frame = tb, framesep = \fboxsep, framexleftmargin = 20pt}

\lstdefinestyle{gostyle}{
    language=python,
    basicstyle=\ttfamily\scriptsize,
    breakatwhitespace=false,
    breaklines=true,
    captionpos=b,
    keepspaces=true,
    showspaces=false,
    showstringspaces=false,
    showtabs=false,
    tabsize=2,
    commentstyle=\color{codegray},
    keywordstyle=\color{blue},
    stringstyle=\color{codepurple},
    moredelim=[is][\color{codered}\textbf]{@@}{@@},
    moredelim=[is][\color{codegreen}\textbf]{^}{^},
    columns=flexible,
    aboveskip=-0.5\baselineskip,
    belowskip=-0.0\baselineskip,
    lineskip=0\baselineskip,
    numbers=left,
    numberstyle=\tiny\color{codegray},
    stepnumber=1,
    firstnumber=auto,
    numbersep=4pt,
    frame=none,
}
\lstset{style=gostyle}

\begin{table}[ht]
\caption{
Failure analysis example 1: false rejection
}
\label{tab:fail_case1}
\centering
\vspace{-5pt}
\small
\tabcolsep=0.06cm
\begin{tabular}{|>{\centering\arraybackslash}m{0.05\linewidth}|p{0.02\linewidth}|p{0.84\linewidth}|}
\hline
\multicolumn{3}{|c|}{\textbf{File:} \texttt{tests/components/fritzbox/test\_binary\_sensor.py}} \\
\hline
\multirow{4}{*}{$h_{-1}$} 
& & \begin{lstlisting}
  from homeassistant.const import (
      ...
\end{lstlisting} \\

 & & \cellcolor{red!20}
\begin{lstlisting}[firstnumber=3]
-     STATE_OFF,
\end{lstlisting} \\
& & \begin{lstlisting}[firstnumber=4]
      STATE_ON,
\end{lstlisting} \\
\cline{1-1}
\multirow{2}{*}{\shortstack{$h'$\\($h_{0}$)}}
& & \cellcolor{green!20}
\begin{lstlisting}[firstnumber=5]
+     STATE_UNAVAILABLE,
\end{lstlisting} \\
& & \begin{lstlisting}[firstnumber=6]
  )
\end{lstlisting} \\
\cline{1-1}
\multirow{5}{*}{$h_{-2}$} & &
\begin{lstlisting}[firstnumber=7]
      ...
      assert state
\end{lstlisting}\\
 & & \cellcolor{red!20}
\begin{lstlisting}[firstnumber=9, frame=none]
-     assert state.state == STATE_OFF
\end{lstlisting} \\
& & \cellcolor{green!20}
\begin{lstlisting}[firstnumber=10, frame=none]
+     assert state.state == STATE_UNAVAILABLE
\end{lstlisting} \\
& &
\begin{lstlisting}[firstnumber=11, frame=none]
      ...
\end{lstlisting} \\
\hline
\end{tabular}
\end{table}
\noindent\textbf{Example. }This example is observed during the digital twin evaluation of Claude combined with \editflow,
where the simulation replays the editing process of a real-world commit\footnote{\url{https://github.com/home-assistant/core/commit/4241c4c}}.
During the simulated editing sequence, the system iteratively provides recommendations based on the evolving edit context.
The relevant contextual information is summarized in Table~\ref{tab:fail_case1}.
Here, $h_{-x}$ denotes the prior edit that has already been applied at the current simulation step,
indexed relative to the current time.
Specifically, $h_{-1}$ is the most recently applied edit, $h_{-2}$ is the edit immediately before it, and so on.
The symbol $h'$ denotes the edit recommended by Claude at the current step,
which corresponds to the expected next edit at the current progress and is denoted as $h_0$.
The recovered partial-order relations can be expressed as:
$h_{-1} \sim h_{-2}$,
$h_{-2}\sim h_0$,
and $h_{-1} \perp h_0$.
Under the recovered graph, $h_0$ is a one-hop successor of $h_{-2}$ and therefore constitutes a flow-keeping (and correct) subsequent edit in the global context.}{(C2,C19)}

\highlight{However, during flow-aware filtering, \keeper adopts a 1-context sensitivity assumption and does not evaluate $\lambda(h, h')$ for all $h \in h_{<0}$.
Instead, it only evaluates the relation $\lambda(h_{-1}, h')$, considering the candidate edit solely with respect to the most recent prior edit.
In this case, \keeper yields $\lambda(h_{-1}, h') = \perp$, as the local flow continuity between $h_{-1}$ and $h'$ is not apparent.
As a result, the critical mental-flow context established by $h_{-2}$ is ignored,
leading \keeper to misclassify $h'$ as non-flow-keeping edit,
and to incorrectly filter out an otherwise correct recommendation.}{(C2,C19)}

\highlight{\noindent\textbf{Metric implications.}
Filtering out a flow-keeping edit directly reduces the proportion of flow-keeping recommendations in the final evaluation.
Although a rejected recommendation may enter the recycling loop and be deferred until a more suitable moment (details may refer to \autoref{sec:optimization}),
in this example, the rejected edit continues to suffer from the $1$-context sensitivity issue discussed above,
namely that its flow relation is only evaluated against the most recent prior edit.
As no subsequent edits occur that would render $h_0$ a one-hop successor of the immediately preceding edit,
the recommendation is eventually discarded,
resulting in a False Negative in flow-graph-independent metrics.
Specifically, a correct recommendation is permanently removed from the candidate set,
leading to simultaneous decreases in both precision and recall.}{(C2,C19)}

\subsection{Failure Mode II: Acceptance of Seemingly Flow-Keeping but Incorrect Edits}
\highlight{\textbf{Mechanism.}
In this scenario, \keeper fails to filter out certain recommendations because it deems them to be flow-keeping based on local flow continuity.
However, the inferred flow continuity does not correspond to the user's intended editing flow, leading \keeper to incorrectly accept edits that are locally coherent but misaligned with the user's actual intent.
This failure typically stems from the ambiguity of editing intentions,
where a single prior edit can logically lead to multiple divergent flow paths based on different intentions.
While \keeper is designed to verify flow continuity, it primarily validates whether a recommendation constitutes \emph{a} plausible continuation,
rather than ensuring it is \emph{the specific} continuation intended by the developer.
When a recommendation aligns with a logical but unintended flow branch (e.g., deletion instead of modification),
\keeper lacks the sufficient intent awareness to distinguish it from the ground truth,
leading to the acceptance of a locally coherent but globally incorrect edit under the developer's intended flow.}{(C2,C19)}

\definecolor{codegreen}{rgb}{0.0, 0.66, 0.30}
\definecolor{codered}{rgb}{1,0.05,0.05}
\definecolor{codegray}{rgb}{0.5,0.5,0.5}
\definecolor{codegray}{rgb}{0.5,0.5,0.5}
\definecolor{codeblue}{rgb}{0.0,0.0,0.5}
\definecolor{codered}{rgb}{0.6,0.0,0.0}
\definecolor{codepurple}{rgb}{0.58,0.0,0.82}
\definecolor{codeorange}{rgb}{1.0,0.5,0.0}
\definecolor{codecyan}{rgb}{0.0,0.6,0.6} 

\lstdefinestyle{lgeometry}{ xleftmargin = 20pt, xrightmargin = 0pt, frame = tb, framesep = \fboxsep, framexleftmargin = 20pt}

\lstdefinestyle{gostyle}{
    language=python,
    basicstyle=\ttfamily\scriptsize,
    breakatwhitespace=false,
    breaklines=true,
    captionpos=b,
    keepspaces=true,
    showspaces=false,
    showstringspaces=false,
    showtabs=false,
    tabsize=2,
    commentstyle=\color{codegray},
    keywordstyle=\color{blue},
    stringstyle=\color{codepurple},
    moredelim=[is][\color{codered}\textbf]{@@}{@@},
    moredelim=[is][\color{codegreen}\textbf]{^}{^},
    columns=flexible,
    aboveskip=-0.5\baselineskip,
    belowskip=-0.0\baselineskip,
    lineskip=0\baselineskip,
    numbers=left,
    numberstyle=\tiny\color{codegray},
    stepnumber=1,
    firstnumber=auto,
    numbersep=4pt,
    frame=none,
}
\lstset{style=gostyle}

\begin{table}[ht]
\caption{
Failure analysis example 2: false acceptance
}
\label{tab:fail_case2}
\centering
\vspace{-5pt}
\small
\tabcolsep=0.06cm
\begin{tabular}{|>{\centering\arraybackslash}m{0.05\linewidth}|p{0.02\linewidth}|p{0.84\linewidth}|}
\hline
\multicolumn{3}{|c|}{\textbf{File:} \texttt{zerver/lib/scheduled\_messages.py}} \\
\hline
\multirow{5}{*}{$h_{-1}$} 
& & \begin{lstlisting}
  from homeassistant.const import (
      ...
\end{lstlisting} \\
 & & \cellcolor{red!20}
\begin{lstlisting}[firstnumber=3]
-     StreamScheduledMessageAPI,
\end{lstlisting} \\
& & \begin{lstlisting}[firstnumber=4]
      UserProfile,
  )
\end{lstlisting} \\
\hline
\multicolumn{3}{|c|}{\textbf{File:} \texttt{zerver/models.py}} \\
\hline
\multirow{5}{*}{$h_{0}$}
& & \cellcolor{red!20}\begin{lstlisting}[firstnumber=6]
- class StreamScheduledMessageAPI(TypedDict):
\end{lstlisting} \\
& & \cellcolor{green!20}
\begin{lstlisting}[firstnumber=7]
+ class APIScheduledStreamMessageDict(TypedDict):
\end{lstlisting} \\
& & \begin{lstlisting}[firstnumber=8]
      scheduled_message_id: int
      to: int
      ...
\end{lstlisting} \\
\hline
\multirow{5}{*}{$h'$} 
 & & \cellcolor{red!20}
\begin{lstlisting}[firstnumber=6]
- class StreamScheduledMessageAPI(TypedDict):
-     scheduled_message_id: int
-     to: int
-     type: str
\end{lstlisting} \\
& &
\begin{lstlisting}[firstnumber=10]
      content: str
\end{lstlisting} \\
\hline
\end{tabular}
\end{table}
\highlight{\noindent\textbf{Example.}
This example\footnote{\url{https://github.com/zulip/zulip/commit/02fafb03}} is observed during the digital twin evaluation of Cursor combined with \editflow.
The relevant edits are available in Table~\ref{tab:fail_case2}.
Here, $h_{-1}$ denotes the most recently applied prior edit at the current simulation step,
$h_0$ denotes the expected ground-truth next edit,
and $h'$ denotes the edit recommended by Cursor at the current step.
In this case, $h_{-1}$ removes the import of \texttt{StreamScheduledMessageAPI}.
The ground-truth next edit $h_0$ then renames the corresponding class in \texttt{zerver/models.py},
changing \texttt{Stream\allowbreak Scheduled\allowbreak Message\allowbreak API} to \texttt{API\allowbreak Scheduled\allowbreak Stream\allowbreak Message\allowbreak Dict}.
However, Cursor recommends $h'$, which instead removes (or substantially deletes) the existing definition.
\keeper identifies the transition from $h_{-1}$ to $h'$ as flow-continuous,
e.g., $\lambda(h_{-1}, h') \in \{\prec, \sim\}$,
since both edits operate on the same API and appear cognitively coherent as a follow-up to removing its import.
As a result, \editflow does not filter out $h'$.}{(C2,C19)}

\highlight{Nevertheless, $h'$ does not align with the developer's actual intent in this commit,
which is to consistently \emph{rename} the API type rather than delete its definition.
This constitutes a false acceptance:
a recommendation that appears flow-keeping under the immediate context,
yet is incorrect and flow-breaking with respect to the human-intended ground-truth edit $h_0$.}{(C2,C19)}

\noindent\textbf{Metric implications.}
\highlight{As $h'$ is accepted as an incorrect recommendation with respect to ground-truth,
it is counted as a false positive and classified as a flow-breaking outcome with respect to the human ground-truth edit sequence.
Consequently, this case increases the proportion of flow-breaking recommendations
and simultaneously lowers precision in flow-graph-independent metrics.
Meanwhile, recall is not directly affected by this false positive.}{(C2,C19)}

\section{Threats to Validity}
We discuss three major types of threats to the validity of our study.

\highlight{\textbf{External validity.}
Our benchmark primarily consists of 100 Python commits and one industrial dataset that cannot be publicly released due to compliance restrictions.
The dataset uses Git commits solely as a data source and does not reflect GitHub-specific workflows such as pull requests or code reviews.
Although prior work suggests that developer editing behaviours are broadly consistent across programming languages \cite{ray2014large, zhang2023multilingual, effendi2023language},
we acknowledge that our data composition may limit the generalizability of our findings to other programming languages, development workflows, or industrial settings.
Validating the mental-flow model on additional languages and non-GitHub workflows would further strengthen external validity.
To mitigate this concern, we selected projects covering diverse functionalities and repository sizes,
and we plan to extend our benchmark to multi-language repositories and additional industrial contexts in future work.}{(C15)}

\textbf{Construct validity.}
The notion of \textit{mental-flow alignment} is inherently abstract, and our operationalization through the \textsc{Keep}/\textsc{Jump}/\textsc{Revert}/\textsc{Break} taxonomy may only approximate developers’ cognitive states. 
Although we employed double annotation and inter-annotator agreement checks to ensure labelling reliability, some subjectivity is inevitable.

\textbf{Internal validity.}
\highlight{First, the digital-twin simulation assumes that developers always make correct decisions, while abstracting away other cognitive and behavioural factors,
such as developers' trust in recommendations and individual programming habits.
These simplifying assumptions reduce the complexity of real developer behaviour and may therefore overestimate the effectiveness of flow-aware optimization.
Nevertheless, this abstraction provides a reproducible and controlled environment for benchmarking heterogeneous agents (e.g., Cursor vs. Claude) by eliminating human variance.}{(C17)}

\highlight{Second, the collected edit-order data, while reflecting real developer actions, does not guarantee the most mental-flow-compatible editing sequence.
Developers may occasionally follow suboptimal or counter-intuitive paths due to habits, local context, or external constraints.
Although our manual inspection suggests such cases are rare, this imperfection may introduce noise into the learned flow model.}{(C14)}

Third, in \textbf{RQ2}, we estimate ordering quality based on the proportion of \textit{violations} against the observed ground-truth edit order. 
This evaluation protocol may introduce an optimistic bias, as some false-positive relations in the predicted partial order cannot be falsified by a single observed editing trajectory. 
Consequently, the measured performance may slightly overestimate the true ordering accuracy.
Nevertheless, this controlled setup enables consistent comparison across systems by eliminating human variance. 

Fourth, our approach relies on large language models for edit order inference, which are inherently stochastic. 
As a result, identical inputs may occasionally yield different inferred partial orders, potentially introducing variability into flow-aware filtering and downstream evaluation. 
While this randomness does not affect the conceptual validity of our approach, it may influence the stability of individual predictions. 
In practice, such variability can be mitigated through standard engineering techniques, such as lowering decoding temperature, applying majority voting across multiple runs, or aggregating predictions over multiple samples. 
We leave a systematic exploration of these stabilization strategies to future work.

Future work will explore incorporating more nuanced cognitive and behavioral factors into the simulation, as well as real-world IDE traces and multi-path simulation, to better approximate realistic developer behavior.

\section{Related Work}

\subsection{Code Edit Recommendation}
Static analysis methods for code edit recommendations primarily rely on extracting edit patterns.
For instance, 
CCDemon \cite{lin2015clone} exploits clone differences for suggesting pasted-code modifications,
Overwatch \cite{overwatch} leverages IDE-logged traces,
Pyevolve \cite{dilhara2023pyevolve} mines frequently repeated code change patterns (CPATs) from the version histories of Python projects,
and AppEvolve \cite{fazzini2019automated} extracts API changes from other applications.
Solutions of this kind prioritize efficiency and syntactic correctness. However, they suffer from poor generalization when extended to more diverse editing scenarios.

Recent advances in large language models have significantly reshaped software engineering, 
enabling applications such as code completion, bug fixing, debugging, refactoring, and automated testing \cite{liu2025guipilot,ren2023api,lin2021graph, lin2020recovering, qi2025intention, cai2025automated}. 
Beyond single-shot tasks, recent work has increasingly explored code edit recommendation,
such as CoditT5 \cite{ZhangETAL22CoditT5}, CODEEDITOR \cite{li2023codeeditor}, and CCT5 \cite{lin2023cct5}, propose pre-trained models specifically designed for code editing tasks.
To further enhance performance, various forms of contextual retrieval have been explored; for instance, GrACE \cite{grace} incorporates relevant prior edits, while SARGAM \cite{liu2024automated} adopts a Search-Generate-Modify paradigm.
Other efforts, including Codit \cite{codit} and Recoder \cite{recoder}, integrate abstract syntax tree (AST) representations to improve the model’s understanding of code structure and syntax.
Moreover, recent work such as CodePlan \cite{bairi2024codeplan} integrates LLMs with static analysis tools, including dependency graphs, to support more accurate reasoning over code changes.
However, these approaches remain largely confined to evaluations on offline benchmarks and have not been designed or validated for project-wise subsequent edit recommendation in realistic development scenarios.

For end-to-end editing scenarios, the most relevant work in the academic community is \coedpilot \cite{code-edit-pilot}, which proposes a Transformer-based framework for jointly localizing and generating subsequent edits.
Building upon this framework, \cite{liu2025learning} incorporates static analysis tools and fine-grained code edit representations to further enhance performance.
Meanwhile, numerous products have been developed in the industry to support similar workflows.
Cursor \cite{cursor} was the first IDE to support proactive subsequent edit suggestions, introduced in November 2023. 
In this design, suggested edits are proactively rendered as ghost text.
Developers can press \kbdbox{Tab} once to navigate to the suggestion and press \kbdbox{Tab} a second time to accept it.
Cursor also supports the command-response paradigm, in which users describe their intentions, 
and the agent actively searches inside the codebase and returns corresponding code edit suggestions along with an explanation.
The above 2 solutions are then quickly adopted by other products, including Copilot \cite{copilot} and Windsurf \cite{windsurf}.
More recently, agentic approaches have emerged that operate directly through command-line interfaces. 
Claude Code \cite{claude-code}, released in late 2024, enables developers to delegate coding tasks to an AI agent that can read, write, and execute code autonomously within a specified project directory.
Other similar products include Gemini CLI \cite{gemini-cli} and Qwen Code \cite{qwen-code}.

\subsection{Developer Flow and Productivity}

\label{subsec:flow_productivity}
\highlight{\textbf{Flow as a core construct in productivity.}
Mental flow is a well-established psychological construct defined as a mental state in which a person performing an activity is fully immersed in a feeling of energized focus, full involvement, and enjoyment~\cite{csikszentmihalyi1990flow}. 
Since its introduction, flow has been widely adopted in software engineering research as a key lens for understanding developer productivity and experience. 
The SPACE framework, now widely recognized as an industry standard for evaluating developer productivity, treats \emph{Efficiency and Flow} as one of its five independent dimensions~\cite{forsgren2021space}. 
Similarly, the DevEx framework ~\cite{noda2023devex} distills developer productivity into three core dimensions: 
feedback loops, cognitive load, and flow state,
positioning flow as a first-class driver of effective software development. Together, these frameworks establish flow not as a subjective afterthought, but as a foundational construct for reasoning about productivity in both academic and industrial contexts.}{(C1,C13)}

\highlight{\textbf{Quantitative relationship between flow and productivity.}
Prior work provides strong quantitative evidence linking sustained flow to productivity gains and interruptions to substantial productivity loss. Meyer \emph{et al.} report that over 50\% of developers associate productive workdays with uninterrupted flow and minimal context switching~\cite{meyer2014software}. Large-scale GitHub and DevEx studies further show that developers who sustain deep, uninterrupted work report 30--50\% higher perceived productivity and up to 20\% higher innovation~\cite{Eirini_GitHub,noda2023devex}.  
Conversely, the cost of breaking flow is disproportionate to the frequency of interruptions. After an interruption, developers require on average 23 minutes and 15 seconds to fully resume their original task (\emph{recovery tax})~\cite{mark2008cost}. Even a modest interruption rate of 10\% can nearly double total task completion time, indicating a non-linear, compounding effect of flow disruption on productivity~\cite{sum2015analysis}.}{(C1,C13)}

\highlight{\textbf{Disruptions of flow by modern coding tools.}
While flow is beneficial, a growing body of evidence shows that modern development tools,
particularly AI-based coding assistants,
can both support and disrupt it. 
Beyond their well-known positive effects, 
prior studies document substantial downsides of LLM-based coding assistants. Unwanted or poorly timed suggestions, verbose outputs, interface switching, and the need to constantly verify generated code introduce frequent interruptions that break developers’ flow continuity~\cite{mohamed2025impact}.
Controlled and observational studies further demonstrate that such interruptions lead to excessive context switching between thinking, reading, and debugging, resulting in increased cognitive load and degraded task performance~\cite{prather2023s}. 
Qualitative evidence from both professional and novice developers highlights frustration, loss of focus, and exhaustion when tools intrude into the developer’s mental flow, for example by flooding the screen with unsolicited suggestions or requiring sustained attention to tool interaction rather than problem solving~\cite{prather2023s,pimenova2025good}. 
These findings align with earlier work showing that interruptions, delays, and tool friction are among the primary barriers to experiencing flow in software engineering tasks~\cite{meyer2014software,janssens2022go}.}{(C1,C13)}

\highlight{Taken together, prior work establishes flow as a central, measurable determinant of developer productivity, quantifies its impact, and provides converging evidence that poorly aligned tooling can significantly disrupt flow.
This body of research motivates the need for developer assistance systems that explicitly preserve cognitive continuity rather than inadvertently fragment it.}{(C1,C13)}

\section{Conclusion}

This paper presents EditFlow, the first framework for benchmarking and optimizing code edit recommendation systems from the perspective of developers' mental flow. 
We identified an essential gap between technical accuracy and developer productivity, 
showing that existing assistants often disrupt developers' cognitive continuity due to their reliance on static commit snapshots.

To bridge this gap, EditFlow integrates three synergistic components: 
a prompt auto-tuning mechanism that reconstructs edit orders, 
a digital twin for flow-aware evaluation, 
and a unified optimization wrapper for filtering and re-ranking recommendations. 
Extensive experiments on annotated and industrial datasets show that EditFlow improves edit-order reconstruction accuracy by 63.81\%, reduces flow violations by 75\%, and increases recommendation precision by 66.99\%. 
A controlled user study with 32 participants further confirms 25.11\% faster task completion and significantly higher perceived recommendation quality.

Overall, EditFlow establishes flow-awareness as a new dimension for evaluating and enhancing AI-assisted code editing. By reconstructing and simulating developers' editing processes, our framework provides a foundation for future research on cognitively aligned recommendation systems and seamless human-AI code collaboration.

\section*{Data-Availability Statement}
The source code, auto-tuned prompt, dataset, and experiment results are available at \cite{homepage}.
Please note that a portion of the dataset provided by our industry collaborator cannot be publicly released due to compliance and confidentiality agreements.

\begin{acks}

This research is supported in part by the National Natural Science Foundation of China (62572300), 
the Ministry of Education, Singapore (MOE-T2EP20124-0017, MOET32020-0004), 
the National Research Foundation, Singapore and the Cyber Security Agency under its National Cybersecurity R\&D Programme (NCRP25-P04-TAICeN), 
DSO National Laboratories under the AI Singapore Programme (AISG Award No: AISG2-GC-2023-008-1B), 
and Cyber Security Agency of Singapore under its National Cybersecurity R\&D Programme and CyberSG R\&D Cyber Research Programme Office. 
Any opinions, findings and conclusions or recommendations expressed in this material are those of the author(s) and do not reflect the views of National Research Foundation, Singapore, Cyber Security Agency of Singapore as well as CyberSG R\&D Programme Office, Singapore.

\end{acks}


\bibliographystyle{ACM-Reference-Format}
\bibliography{Reference}



\end{document}